\documentclass[11pt]{article}
\usepackage[margin=1in]{geometry}
\usepackage{fullpage}
\usepackage{tgtermes}
\usepackage[T1]{fontenc}
\usepackage[colorlinks,citecolor=blue,linkcolor=blue,urlcolor=black]{hyperref}
\usepackage{graphicx}
\usepackage{mathtools}
\usepackage{amsfonts,amsmath,amsthm,amssymb,dsfont}
\usepackage{xfrac,nicefrac}
\usepackage{mathdots}
\usepackage{bm,bbm}
\usepackage{url}
\usepackage{paralist}
\usepackage{enumerate}
\usepackage[normalem]{ulem}
\usepackage{xspace}
\xspaceaddexceptions{]\}}
\usepackage{cleveref}
\usepackage{comment}
\usepackage{tabu}
\usepackage{framed}
\usepackage{float,wrapfig}
\usepackage{tikz}
\usepackage[usenames,dvipsnames]{pstricks} 

\usepackage{enumitem}
\usepackage{adjustbox}
\usepackage{algorithm}
\usepackage{algpseudocode}

\algnewcommand{\algorithmicand}{\textbf{AND }}
\algnewcommand{\algorithmicor}{\textbf{OR }}
\algnewcommand{\algorithmicxor}{\textbf{XOR }}
\algnewcommand{\algorithmicnot}{\textbf{NOT }}
\algnewcommand{\OR}{\algorithmicor}
\algnewcommand{\AND}{\algorithmicand}
\algnewcommand{\XOR}{\algorithmicxor}
\algnewcommand{\NOT}{\algorithmicnot}
\algnewcommand{\var}{\texttt}

\algnewcommand{\algorithmicbreak}{\textbf{break}}
\algnewcommand{\Break}{\algorithmicbreak}

\usetikzlibrary{hobby}
\usetikzlibrary{decorations.markings}
\usetikzlibrary{quotes,angles}

\DeclarePairedDelimiter\abs{\lvert}{\rvert}%
\DeclarePairedDelimiter\norm{\lVert}{\rVert}%

\makeatletter
\let\oldabs\abs
\def\abs{\@ifstar{\oldabs}{\oldabs*}}
\let\oldnorm\norm
\def\norm{\@ifstar{\oldnorm}{\oldnorm*}}
\makeatother

\theoremstyle{plain}

\newtheorem{theorem}{Theorem}[section]
\newtheorem{lemma}[theorem]{Lemma}

\newtheorem{cor}[theorem]{Corollary}

\theoremstyle{definition}
\newtheorem{definition}[theorem]{Definition}

\renewcommand{\epsilon}{\varepsilon}
\newcommand{\eps}{\varepsilon}

\renewcommand{\tilde}{\widetilde}
\renewcommand{\hat}{\widehat}

\newcommand{\ed}{\mathsf{ed}}
\newcommand{\size}{\mathsf{size}}

\newcommand{\parent}{\mathsf{par}}
\newcommand{\child}{\mathsf{child}}
\newcommand{\similarity}{\mathsf{sim}}
\newcommand{\sub}{\mathsf{sub}}
\newcommand{\childcount}{\mathsf{d}}
\newcommand{\rangemax}{\mathsf{rangemax}}
\newcommand{\Left}{\mathsf{l}}
\newcommand{\Right}{\mathsf{r}}
\newcommand{\maxrow}{\mathsf{maxrow}}
\newcommand{\mincol}{\mathsf{mincol}}
\newcommand{\Root}{\mathsf{root}}
\newcommand{\vroot}{\mathsf{vroot}}
\newcommand{\MUL}{\mathsf{MUL}}

\usepackage{caption}
\usepackage{transparent}
\usepackage{subcaption}

\usetikzlibrary{shapes.geometric}
\usetikzlibrary{arrows}

\tikzset{
triangle/.style={
  draw,solid,isosceles triangle,shape border rotate=90},
}

\title{Breaking the Cubic Barrier for (Unweighted) Tree Edit Distance}

\author{
  Xiao Mao \\
  Massachusetts Institute of Technology \\
  \texttt{xiao\_mao@mit.edu} \\
}

\date{}

\begin{document}

\maketitle

\begin{abstract}
The (unweighted) \emph{tree edit distance} problem for $n$ node trees asks to compute a measure of dissimilarity between two rooted trees with node labels. The current best algorithm from more than a decade ago runs in $O(n ^ 3)$ time [Demaine, Mozes, Rossman, and Weimann, ICALP 2007]. The same paper also showed that $O(n ^ 3)$ is the best possible running time for any algorithm using the so-called \emph{decomposition strategy}, which underlies almost all the known algorithms for this problem. These algorithms would also work for the \emph{weighted} tree edit distance problem, which cannot be solved in truly sub-cubic time under the APSP conjecture [Bringmann, Gawrychowski, Mozes, and Weimann, SODA 2018].

In this paper, we break the cubic barrier by showing an $O(n ^ {2.9546})$ time algorithm for the \emph{unweighted} tree edit distance problem.

We consider an equivalent maximization problem and use a dynamic programming scheme involving matrices with many special properties. By using a decomposition scheme as well as several combinatorial techniques, we reduce tree edit distance to the max-plus product of bounded-difference matrices, which can be solved in truly sub-cubic time [Bringmann, Grandoni, Saha, and Vassilevska Williams, FOCS 2016].
\end{abstract}

\section{Introduction}
\label{sec:intro}

\subsection{Overview}

One of the most fundamental problems in computer science is the \emph{(string) edit distance} problem, studied since the 1960's. Defined as the minimum number of deletions, insertions or substitutions needed to change one string into another, it is a natural way to measure the dissimilarity between data that can be represented as strings. As a measure for linearly ordered data, edit distance is not as useful when the data is hierarchically organized. For data that is in the form of ordered trees, \emph{tree edit distance} serves as a natural generalization of edit distance. First introduced by Selkow in 1977 \cite{selkow77}, it has found applications in a variety of areas such as computational biology \cite{gusfield_1997,10.1093/bioinformatics/6.4.309,HochsmannTGK03,waterman1995introduction}, structured data analysis \cite{KochBG03,Chawathe99,FerraginaLMM09}, image processing \cite{BellandoK99,KleinTSK00,KleinSK01,SebastianKK04}, and compiler optimization \cite{demaine2007}. One of the most notable applications is the analysis of RNA molecules whose secondary structures are typically represented as rooted trees \cite{gusfield_1997,doi:10.1137/0213024}.

For two \emph{rooted ordered} trees whose nodes are labeled with symbols, their tree edit distance is the minimum number of node deletions, insertions, and relabelings needed to change one tree into the other. When a node is deleted, its children become children of its parent. The resulting trees must be structurally identical, where the order of siblings matters, and symbols on corresponding nodes must also match.

\paragraph{Previous results.} Prior to our work, there has been a long line of research producing efficient algorithms for tree edit distance, as shown in Table \ref{table:ted}. In 1979, Tai \cite{Tai79} gave the first algorithm that computes this metric for two trees of size $n$ in $O(n ^ 6)$ time. In 1989, the running time was improved to $O(n^4)$ by Zhang and Shasha \cite{ZhangS89}, notably by using a dynamic programming approach. Their approach served as the basis for later algorithms. In 1989, Klein \cite{Klein98} obtained an $O(n^3 \log n)$ time algorithm by adapting a better strategy in deciding the direction of transitions in Zhang and Shasha's dynamic programming scheme and analyzing the running time using heavy-light decomposition. Finally, Demaine, Mozes, Rossman, and Weimann \cite{demaine2007} improved the running time to $O(n^3)$ by further optimizing this dynamic programming scheme and showed that their running time is the theoretical lower bound among a certain class of dynamic programming algorithms termed \emph{decomposition strategy algorithms} by Dulucq and Touzet \cite{DulucqT03,DulucqT05}. 
\begin{table}[h]
\begin{center}
\begin{tabular}{ |c|c|c| }
 \hline
 Runtime & Authors & Year\\
 \hline
 $O(n ^ 6)$   & Tai \cite{Tai79}   & 1979 \\
 $O(n ^ 4)$   & Shasha and Zhang \cite{ZhangS89} & 1989 \\
 $O(n ^ 3 \log n)$   & Klein \cite{Klein98} & 1998 \\
 $O(n ^ 3)$   & Demaine, Mozes, Rossman, and Weimann \cite{demaine2007} & 2007 \\
 $O(n ^ {2.9546})$ & \textbf{This paper} & \\
 \hline
\end{tabular}
\end{center}
\caption{Tree edit distance algorithms} \label{table:ted}
\end{table}

Another algorithm that is not so frequently mentioned in tree edit distance literature is the algorithm by Chen in 2001 \cite{chen01}, which is efficient when the number of leaves in one of the trees is small. Despite not being as well-known, Chen's algorithm is similar to our starting algorithm and will be discussed in detail in Section \ref{sec:review}.

All these previous algorithms also compute tree edit distance in a \emph{weighted} setting where the costs of deletions, insertions, or relabeling are not necessarily $1$, but functions of the symbols on the nodes involved. For the weighted case, Bringmann, Gawrychowski, Mozes, and Weimann \cite{apsphard2020} showed that a truly sub-cubic\footnote{By truly sub-cubic we mean $O(n ^ {3 - \eps})$ for some constant $\eps > 0$.} time algorithm would imply a truly sub-cubic time algorithm for the All-Pairs Shortest Paths (APSP) problem (assuming alphabet of size $\Theta(n)$), as well as an $O(n ^ {k(1 - \eps)})$ time algorithm for the Max-weight $k$-clique problem (assuming a sufficiently large constant-size alphabet). It is conjectured that neither of these algorithms exists \cite{vvwsurvey}, and therefore a truly sub-cubic algorithm for the weighted case is unlikely.

It is interesting to point out that the landscape for the \emph{string} edit distance problem now drastically differs from that of the tree edit distance problem: there is a quadratic-time fine-grained lower bound for the string edit distance problem based on the Strong Exponential Time Hypothesis (SETH) that also holds for the unweighted case (i.e. unit-cost operations)  \cite{BackursI18,AbboudBW15}. Sub-polynomial improvement has already been found for the unweighted case and the current best running time is $O(n ^ 2/\log ^ 2 n)$ for a finite alphabet \cite{MasekP80} and $O(n ^ 2 (\log \log n) ^ 2 / \log ^ 2 n)$ for an arbitrary alphabet \cite{tcs/BilleF08}. 

\subsection{Main result}

Our main result is the following:
\begin{theorem} \label{theorem:main}
    There is an $O(nm ^ {1.9546})$ time randomized algorithm and an $O(nm ^ {1.9639})$ time deterministic algorithm that computes the (unweighted) tree edit distance between two trees of sizes $n$ and $m$. 
\end{theorem}
When $m = O(n)$, this implies an $O(n ^ {2.9546})$ time randomized and an $O(n ^ {2.9639})$ time deterministic algorithm that solve the tree edit distance problem, which are the first ever known truly sub-cubic algorithms for this problem.

The two exponents in our results are from applications of the work by Bringmann, Grandoni, Saha, and Vassilevska Williams \cite{Bringmann16}, which shows that the max-plus product of two \emph{bounded-difference} $n \times n$ matrices (i.e. bounded difference between adjacent entries, see Definition \ref{def:bd}) can be computed in $O(n ^ {2.8244})$ randomized and $O(n ^ {2.8603})$ deterministic time. 



Although we focus our discussion on the unweighted setting, our algorithm does not actually require the weights of the operations to be exactly $1$. It can be easily verified that our algorithm would also work for weighted tree edit distance where weights are integers bounded by $W$ with a running time polynomially dependent on $W$ and with the same exponents on $n$.

\subsection{Related research}
The tree edit distance problem in the bounded distance setting and approximation setting has received a lot of attention recently. New work in these areas is all under the unweighted setting.

In the bounded tree edit distance problem, we assume that the distance is upper bounded by a given parameter $k$. The motivation is that in many practical scenarios, the tree edit distance tends to be small. One of the previous results for the bounded case was due to Touzet \cite{Touzet05,Touzet07}, who gave an algorithm that runs in $O(nk ^ 3)$ time. Very recently, Akmal and Jin \cite{akmal2021faster} gave an algorithm that runs in $O(nk ^ 2\log n)$ time, which is faster when $k = \omega(\log n)$. For the \emph{string} edit distance problem, when the distance is at most $k$, the current best algorithm with running time $\tilde O(n+k^2)$\footnote{By $\tilde O$ we hide $O(n ^ {o(1)})$ factors such as $\log {n}$.} was due to Myers \cite{Meyers86}, Landau and Vishkin \cite{LandauV88}, who used suffix tree to improve upon an earlier $O(nk)$ time algorithm by Ukkonen \cite{Ukkonen85}.

For the approximation problem, the most notable recent result was the $(1 + \eps)$-approximation algorithm by Boroujeni, Ghodsi, Hajiaghayi, and Seddighin \cite{apxfocs2019}, which runs in $\tilde O(\eps^{-3} n^2)$ time and can be improved to $\tilde O(\eps^{-3} nk)$ if the distance is upper bounded by $k$. Prior to their research, literature on approximation of tree edit distance was scarce, but approximation of \emph{string} edit distance had found much popularity \cite{AndoniO12,AndoniKO10,BoroujeniEGHS18,ChakrabortyDGKS18,BrakensiekR20,KouckyS20}, which ultimately led to a constant-factor approximation algorithm running in near-linear time \cite{AndoniN20}. 

There are other variants of the tree edit distance problem. Some examples are those defined on unrooted or unordered trees, or parameterized by the depth or the number of leaves. For a thorough review, one can refer to the comprehensive survey by Bille \cite{Bille05survey}.

\subsection{Technique overview} \label{sec:approach}

\paragraph{An equivalent maximization problem and why it matters.} The most fundamental idea in our approach is to consider the ``inverted'' maximization problem. Note that the tree edit distance is bounded by the sum of the sizes of the two input trees. We define the non-negative \emph{similarity} to be the the difference between the sum of the sizes of the two trees and the tree edit distance between them: $\similarity(T_1, T_2) = \abs{T_1} + \abs{T_2} - \ed(T_1, T_2)$. We now show the motivation behind this idea by considering the \emph{string} edit distance and similarity. Similarity between strings is analogous to Longest Common Subsequence (LCS), but two matching positions contribute $2$ to the answer if the characters are the same, and $1$ to the answer if the characters differ. 

For two strings $s_1$ and $s_2$, let $\abs{s}$ be the length of the string $s$. Let $s_2[l, r) (1 \le l \le r \le \abs{s_2} + 1)$ be the substring of $s_2$ from the $l$-th character to the $(r - 1)$-th character. We consider two $(\abs{s_2} + 1) \times (\abs{s_2} + 1)$ matrices: the \emph{edit distance matrix} $A(s_1) = a_{ij}(s_1)$, and the \emph{similarity matrix} $B(s_1) = b_{ij}(s_1)$, where
\begin{align*}
    a_{ij}(s_1) = 
    \begin{cases}
        \ed(s_1, s_2[i, j))         & \text{if } i \le j \\
        \infty                    & \text{if } i > j 
    \end{cases}, 
    \qquad
    b_{ij}(s_1) = 
    \begin{cases}
        \similarity(s_1, s_2[i, j))             & \text{if } i \le j \\
        -\infty                               & \text{if } i > j 
    \end{cases}.
\end{align*}
The reason why we consider these matrices is because it is easy to see that for another string $s_3$, let $s_1 + s_3$ be the concatenation of $s_1$ and $s_3$ and we have
\begin{align*}
    a_{ij}(s_1 + s_3) &= \min_{k}{\{a_{ik}(s_1) + a_{kj}(s_3)\}}, \\ 
    b_{ij}(s_1 + s_3) &= \max_{k}{\{b_{ik}(s_1) + b_{kj}(s_3)\}}.
\end{align*}
These transitions take the form of min/max-plus products. Although it is conjectured that in general these products cannot be computed in truly sub-cubic time, more efficient algorithms are possible for matrices with certain special properties. Thus this leads us into a hopeful direction as long as the matrices involved have useful special properties. Consider the following example:
\begin{align*}
    s_1 = \textrm{``abac''}, 
    \qquad
    s_2 = \textrm{``acdca''}.
\end{align*}
We have
\begin{align*}
    A(s_1) = 
        \begin{pmatrix}
            4 & 3 & 2 & 3 & 2 & 3 \\
              & 4 & 3 & 4 & 3 & 4 \\
              & &   4 & 4 & 3 & 3 \\
              &&&       4 & 3 & 3 \\
              &\infty& &  & 4 & 3 \\
              & & &       &   & 4
        \end{pmatrix},
        \qquad
    B(s_1) = 
        \begin{pmatrix}
            0 & 2 & 4 & 4 & 6 & 6 \\
              & 0 & 2 & 2 & 4 & 4 \\
              & &   0 & 1 & 3 & 4 \\
             &&&        0 & 2 & 3 \\
              &-\infty && & 0 & 2 \\
              & & &       &   & 0
        \end{pmatrix}.
\end{align*}
We can see that the similarity matrix $B(s_1)$ has the following properties:
\begin{itemize}
    \item Each row of $B(s_1)$ is monotonically non-decreasing from left to right.
    \item Each column of $B(s_1)$ is monotonically non-increasing from top to bottom.
\end{itemize}
Those properties do not apply to the edit distance matrix $A(s_1)$, which makes it much less convenient to deal with in its plain form. Moreover, for any two strings $s_x$ and $s_y$, from the analogy of similarity to LCS, one can see that similarity matrices have the following property:
\begin{itemize}
    \item The non-$(-\infty)$ entries in the similarity matrix are bounded by $2\min(\abs{s_x}, \abs{s_y})$.
\end{itemize}
The non-$\infty$ entries of the edit distance matrix, however, are only bounded by $\max(\abs{s_x}, \abs{s_y})$, which is greater than $2\min(\abs{s_x}, \abs{s_y})$ when one of the strings is significantly shorter.

These properties carry over to the tree edit distance problem. We extend these matrices to the problem of computing the edit distance between two trees $T_1$ and $T_2$. Let $m = \abs{T_2}$. The entries of a similarity matrix for a forest $F$ are the similarities between $F$ and \emph{subforests} of $T_2$, which are defined using a form of depth-first traversal sequence on $T_2$. These similarity matrices are $(2m + 1) \times (2m + 1)$ matrices, and they still inherit the monotone and bounded properties of similarity matrices on strings, and we shall see in Section \ref{sec:review} that the transitions are in the form of max-plus product even on trees.

We utilize the monotone and bounded properties of the matrices and show that the product of the similarity matrices for two forests $F_1$ and $F_2$ can be computed in $\tilde O(\abs{F_1}\abs{F_2}m)$ time using a simple combinatorial algorithm and some data structure to store and modify the matrices implicitly. Based on a dynamic programming scheme similar to the one in Chen's algorithm \cite{chen01}, we develop a cubic time algorithm which computes the similarity matrix associated with the subtree of every node of the first tree.

Since our cubic algorithm involves max-plus products of similarity matrices, we actually utilize a more generalized type of transitions unseen in the line of algorithms originating from Zhang and Shasha's algorithm, but used to an extent in Chen's algorithm. This is an important reason why we are eventually able to reach a truly sub-cubic running time.

To speed-up our algorithm, a more efficient algorithm for the max-plus product is required. One might think that the similarity matrices involved would admit the \textit{Anti-Monge property} since these matrices do admit this property when defined on strings \cite{doi:10.1137/0219066, doi:10.1137/S0097539795288489}. The famous SMAWK algorithm can compute the max-plus product between two $n \times n$ Anti-Monge matrices in $O(n ^ 2)$ time \cite{10.1145/10515.10546}. However, we will show in Appendix \ref{appendix:notmonge} that the Anti-Monge property does not hold for similarity matrices defined on trees. Our algorithm exploits a different property.

\paragraph{Reduction to bounded-difference matrix multiplications using a decomposition scheme.} We notice that the similarity matrix $B$ in the previous string example has the following property:
\begin{itemize}
    \item Adjacent non-$(-\infty)$ values in $B$ differ by at most 2.
\end{itemize}
Lemma \ref{lemma:bd} shows that this is true for similarity matrices in general even on trees. Thus it is promising to optimize the cubic time of the algorithm to truly sub-cubic by adopting the truly sub-cubic time algorithm of max-plus product between bounded-difference matrices in \cite{Bringmann16}. However, this cannot be done in a straightforward way. As we shall see, it seems impossible to break the cubic barrier if we compute the similarity matrix associated with the subtree of \emph{every} node of the first tree, as is the case in the cubic algorithm. In our main algorithm, we decompose the transitions into blocks to skip some of these nodes. This \emph{decomposition scheme} is the heart of our algorithm.

Let the size of the first tree be $n$. We set a block size $\Delta = n ^ d$ for some $d$ slightly smaller than $0.5$ and we decompose the problem into $O(n / \Delta)$ transitions that either involve concatenating two subforests (type I), which is equivalent to taking the max-plus product of their similarity matrices, or are from a subforest to another subforest that contains it with only $O(\Delta)$ more nodes (type II). 

For the type I transitions, we cannot simply call the algorithm in \cite{Bringmann16} as a sub-routine since when $m = n$, $(n / \Delta) m ^ {2.8244} > n ^ 3$. Recall that multiplying the similarity matrix for two forests $F_1$ and $F_2$ can be computed in $\tilde O(\abs{F_1}\abs{F_2}m)$ time. It turns out that if this can be improved to $O(\abs{F_1} ^ {1 - \eps}m ^ 2)$ time for some $\eps > 0$, the total running time for all these transitions becomes truly sub-cubic. To do this, we first design an efficient sub-routine for max-plus products between matrices where no $-\infty$ below the main diagonal is present. By combining this sub-routine with the algorithm in \cite{Bringmann16}, we develop a recursive algorithm that achieves the desired time bound for similarity matrices.

Many of our core ideas are involved in the algorithm for the type II transitions. The main part of the algorithm is a three-part combinatorial method combining many techniques and can be sped up using the monotone and bounded properties of the matrices to $\tilde O(n\Delta ^ 4)$, which would already imply a truly sub-cubic total running time for our entire algorithm. We are able to further optimize the running time to $\tilde O(n\Delta ^ 3)$ using some data structure for path modifications on trees such as a \emph{link/cut tree} \cite{SLEATOR1983362}. The non-combinatorial part involves max-plus product of similarity matrices, which can be sped up using the algorithm for the type I transitions so that its contribution to the total running time is also truly sub-cubic.

\subsection{Organization} \label{sec:organization}
In Section \ref{sec:prelim}, we introduce the notation used throughout the rest of the paper and formally define tree edit distance. We also formally introduce the results on max-plus products of bounded-difference matrices in \cite{Bringmann16} as well as a computation model for implicitly storing and modifying the row-monotone and column-monotone matrices. In Section \ref{sec:review}, we review the line of classic algorithms originating from Zhang and Shasha's dynamic programming scheme and show why they could not be improved to truly sub-cubic running time. We then review Chen's algorithm and relate it to our algorithm. In Section \ref{sec:algorithm} we introduce our algorithm in detail. Finally in Section \ref{sec:futurework}, we first discuss how our algorithm can potentially be sped up, and then discuss what future work can be done in the area of tree edit distance.

\section{Preliminaries}
\label{sec:prelim}

\subsection{Tree edit distance related definitions}
The tree edit distance problem involves ordered trees. For an ordered tree $T$, each node of $T$ is labeled with a symbol from some given alphabet $\Sigma$. In this paper, we consider the size of the alphabet to be $O(n)$. For a node $u \in T$, let $\parent(u)$ denote the parent node of $u$. Let $\Root(T)$ be the root of the tree $T$.

We treat a forest $F$ as ordered as well, meaning that the order between the trees in the forests is important. Under this setting, we can treat a forest $F$ as a tree with a \emph{virtual root}, and let $\parent(v)$ be the virtual root if $v$ is the root of a tree in the forest. Let $\sub(u)$ be the subtree of $u$. We let $L_F$ denote the leftmost tree in $F$ and $R_F$ denote the rightmost tree in $F$. For a tree $T \in F$, let $F - T$ be the forest we get by removing the tree $T$ from $F$ while keeping the order of the remaining trees. Let $\abs{F}$ be the number of nodes in $F$. For a sequence of forests $F_1, F_2, \cdots, F_k$, let $F_1 + F_2 + \cdots + F_k$ be the concatenation of these forests from left to right. For two forests $F_1, F_2$ of $F$, we say $F_1 \subset F_2$ if all nodes in $F_1$ are in $F_2$. When $F_1 \subset F_2$, we use $F_2 \backslash F_1$ to denote the set of nodes in $F_2$ that are not in $F_1$. An empty forest is denoted by $\emptyset$.

For a node $u$ in a forest $F$, let $\childcount(u)$ be the number of children of $u$. For $1 \le k \le \childcount(u)$, let $\child(u, k)$ be the $k$-th child of $u$ from left to right. Let $\sub(u, x) = \sub(\child(u, x))$ and $\sub(u, [x, y]) = \sub(u, x) + \sub(u, x + 1) + \cdots + \sub(u, y)$. Specifically, for the virtual root $r$ of the forest $F$, let $\childcount(r)$ be the number of trees in $F$ and $\child(r, k)$ be the root of the $k$-th tree from left to right, and define $\sub(r, x)$ and $\sub(r, [x, y])$ similarly.

The \emph{node removal operation} removes a node $v$ from the forest $F$, and let the children of $v$ become children of $\parent(v)$, with the same ordering. The result of the removal is denoted by $F - v$.

\begin{definition}[(Unweighted) Tree Edit Distance]
\label{def:unrooted-ted}
For two forests $F_1$ and $F_2$, we consider the following two types of operations:
\begin{itemize}
    \item Relabeling: changing the label of a node to another symbol in $\Sigma$.
    \item Deletion: using the node removal operation to remove a node.
\end{itemize}

The \textit{tree edit distance} between $F_1$ and $F_2$, denoted by $\ed(F_1,F_2)$, is the minimum number of operations we can perform on $F_1$ and $F_2$ so that they become identical forests.
\end{definition}
Some literature distinguishes between \emph{forest edit distance} and tree edit distance, but for simplicity we will not make this distinction.

Figure \ref{fig:ed} gives an example of tree edit distance between two trees $T_1$ and $T_2$.
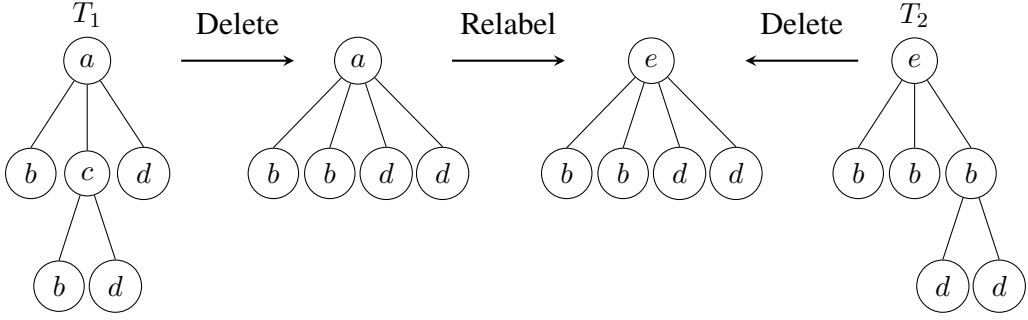
\begin{figure}[H]
    \centering 
    \begin{tikzpicture}[
scale=0.5,
every node/.style={solid},
edge from parent/.style={draw,solid}, 
level/.style={sibling distance=15mm},
level distance=30mm
]
\node [circle,draw,label=above:{$T_1$}] at (0, 0) {$a$}
    [child anchor=north]
    child {node [circle,draw] {$b$} edge from parent[solid]
    }
    child {node [circle,draw] {$c$} edge from parent[solid]
        child {node [circle,draw] {$b$} edge from parent[solid]
        }
        child {node [circle,draw] {$d$} edge from parent[solid]
        };
    }
    child {node [circle,draw] {$d$} edge from parent[solid]
    };
    
\node [circle,draw] at (7.2, 0) {$a$}
    [child anchor=north]
    child {node [circle,draw] {$b$} edge from parent[solid]
    }
    child {node [circle,draw] {$b$} edge from parent[solid]
    }
    child {node [circle,draw] {$d$} edge from parent[solid]
    }
    child {node [circle,draw] {$d$} edge from parent[solid]
    };
    
\node [circle,draw] at (15, 0) {$e$}
    [child anchor=north]
    child {node [circle,draw] {$b$} edge from parent[solid]
    }
    child {node [circle,draw] {$b$} edge from parent[solid]
    }
    child {node [circle,draw] {$d$} edge from parent[solid]
    }
    child {node [circle,draw] {$d$} edge from parent[solid]
    };
    
\node [circle,draw,label=above:{$T_2$}] at (22, 0) {$e$}
    [child anchor=north]
    child {node [circle,draw] {$b$} edge from parent[solid]
    }
    child {node [circle,draw] {$b$} edge from parent[solid]
    }
    child {node [circle,draw] {$b$} edge from parent[solid]
        child {node [circle,draw] {$d$} edge from parent[solid]
        }
        child {node [circle,draw] {$d$} edge from parent[solid]
        };
    };
    \draw [thick, -stealth] (2.5,0) -- (5.5,0);
    \node at (4, 1) {\large Delete};
    \draw [thick, -stealth] (9.7,0) -- (12.7,0);
    \node at (11.2, 1) {\large Relabel};
    \draw [thick, stealth-] (17.5,0) -- (20.5,0);
    \node at (19, 1) {\large Delete};
\end{tikzpicture}
\caption{An optimal series of operations to make $T_1$ and $T_2$ identical is shown, and $\ed(T_1, T_2) = 3$.} \label{fig:ed}
\end{figure}  

We now consider an equivalent maximization problem defined on similarity, which uniquely determines the edit distance:
\begin{definition}[Similarity]
    The \emph{similarity} between two forests $F_1$ and $F_2$ is defined as $\similarity(F_1, F_2) = \abs{F_1} + \abs{F_2} - \ed(F_1, F_2)$.
\end{definition}
Since it is obvious that $\ed(F_1, F_2) \le \abs{F_1} + \abs{F_2}$, similarity is always non-negative.

\begin{definition}[Bi-order traversal sequence]
Consider the depth-first traversal of a forest $F$ starting from the virtual root, with subtrees recursively traversed from left to right. From that we can generate an \emph{bi-order traversal sequence} of length $2\abs{F}$, where each node appears twice, in the following way:
\begin{itemize}
    \item Start from the empty sequence.
    \item Every time we enter or leave a node, we attach the node to the end of the sequence (do not attach the virtual root).
\end{itemize}
We use $F(i)$ to denote the $i$-th node in such sequence.
\end{definition} 

\begin{definition} [Subforest]
For $1 \le l \le r \le 2\abs{F} + 1$, we use $F[l, r)$ to denote the forest obtained by removing from $F$ all nodes that appear at least once in $F(1), F(2), \cdots, F(l - 1)$ or $F(r), F(r + 1) \cdots F(2\abs{F})$, and we call such forest a \emph{subforest} of $F$ (as later illustrated in Figure \ref{fig:synchronous}). 
\end{definition}

For a node $u$, $\Left(u)$ equals the first index where $u$ appears in the bi-order traversal sequence and $\Right(u)$ equals \textit{one plus} the second index where $u$ appears in the sequence. By these definitions, we can see that $F[l, r)$ contains a node $u$ if and only if $l \le \Left(u)$ and $\Right(u) \le r$ and that $F[\Left(u), \Right(u))$ is equal to $\sub(u)$. Note that $F[l, r)$ and $F[l ^ {\prime}, r ^ {\prime})$ might be the same forest for distinct pairs $(l, r)$ and $(l ^ {\prime}, r ^ {\prime})$.

We now introduce the definition of ``synchronous subforests,'' named in a similar spirit to the term ``synchronous decomposition'' from the recent $(1 + \eps)$-approximation paper \cite{apxfocs2019}, which decomposes the tree to disconnected components that become connected after adding one more node and its incident edges.
\begin{definition}[Synchronous subforest] \label{def:syn}
    For a forest $F$, a subforest $F ^ {\prime}$ of $F$ is a \emph{synchronous subforest} of $F$ if there exists a node $u$ that is either a node in $F$ or the virtual root of $F$, and $1 \le x \le y \le \childcount(u)$ such that $F ^ {\prime} = \sub(u, [x, y])$.
\end{definition} 
We can see that the node $u$ works as the \emph{virtual root} of $F ^ {\prime}$, and we denote $u$ as $\vroot(F ^ {\prime})$. Figure \ref{fig:synchronous} shows the bi-order traversal sequence of tree $T$. For node $4$, we have $\Left(4) = 5$ and $\Right(4) = 11$. It also shows a synchronous subforest $F = T[3, 15)$ whose virtual root is node 3. Note that node 2 and 3 are not in $F(3, 15)$ since one of their occurrences is not inside the highlighted interval. We can see that this definition is not a one-on-one mapping, since for example $F(5, 15)$ would be identical to $F(3, 15)$.
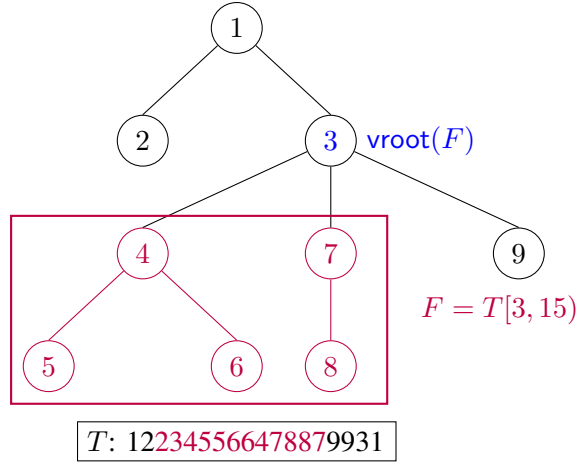
\begin{figure}[H]
    \centering 
    \begin{tikzpicture}[
scale=1,
every node/.style={solid},
edge from parent/.style={draw,solid}, 
level/.style={sibling distance=25mm},
level distance=15mm
]
\node [circle,draw] {$1$}
    [child anchor=north]
    child {node [circle,draw] {$2$} edge from parent[solid]
    }
    child {node [circle,draw,label=right:{\blue $\vroot(F)$}] {\blue $3$} edge from parent[solid]
        child {node [circle,draw, purple] {$4$} edge from parent[solid]
            child {node [circle,draw, purple] {$5$} edge from parent[solid, purple]
            }
            child {node [circle,draw, purple] {$6$} edge from parent[solid, purple]
            };
        }
        child {node [circle,draw, purple] {$7$} edge from parent[solid]
            child {node [circle,draw, purple] {$8$} edge from parent[solid, purple]
            };
        }
        child {node [circle,draw] {$9$}
        };
    };
    \node[draw] at (0,-5.5) {$T$: 12\color{purple}234556647887\color{black}9931};
    \draw[thick, purple] (-3, -2.5) -- (-3, -5) -- (2, -5) -- (2, -2.5) -- cycle;
    \node[purple] at (3.5,-3.75) {$F = T[3, 15)$};
\end{tikzpicture}
\caption{The bi-order traversal sequence and a synchronous subforest of tree $T$} \label{fig:synchronous}
\end{figure}

For two nodes $u$ and $v$ labeled with symbols, let $\delta(u, v)$ be equal to $0$ if $u$ and $v$ have the same symbols and $1$ if their symbols differ. Let $\eta(u, v) = 2 - \delta(u, v)$.

For two nodes $u$ and $v$ such that neither is the ancestor of the other, we have either $\Right(u) \le \Left(v)$ or $\Right(v) \le \Left(u)$, and we say $u$ precedes $v$ if $\Right(u) \le \Left(v)$. For example, node $5$ precedes node $6$ in the tree $T$ in Figure \ref{fig:synchronous}.

\paragraph{Mapping.}
The maximization of the similarity between $F_1$ and $F_2$ can be interpreted as finding a mapping of maximum weight. The mapping we use is identical to the mapping used in Section 2.2 of \cite{ZhangS89}, except for the fact that we are considering similarity. 

The mapping is between two sequences of distinct nodes $\{u_1, u_2, \cdots, u_k\} \in {V(F_1)} ^ k$, $\{v_1, v_2, \cdots, v_k\} \in {V(F_2)} ^ k$, such that for all $1 \le i < j \le k$,
\begin{itemize}
    \item $u_i$ is an ancestor of $u_j$ in $T_1$ if and only if $v_i$ is an ancestor of $v_j$ in $T_2$,
    \item $u_j$ is an ancestor of $u_i$ in $T_1$ if and only if $v_j$ is an ancestor of $v_i$ in $T_2$, and
    \item If neither of $u_i$ and $u_j$ is the ancestor of the other, $u_i$ precedes $u_j$ in $T_1$ if and only if $v_i$ precedes $v_j$ in $T_2$.
\end{itemize}
For each $i$ we map $u_i$ to $v_i$, and the weight of the mapping is
\begin{align*}
    \sum_{1 \le i \le k}{\eta(u_i, v_i)}.
\end{align*}
A node $u \in T_1$ is \emph{mapped} if $u \in \{u_1, u_2, \cdots u_k\}$.

Figure \ref{fig:mapping} shows the mapping of maximum weight between the same $T_1$ and $T_2$ as in Figure \ref{fig:ed}. Nodes with the same subscripts are mapped to each other. We have $\similarity(T_1, T_2) = \abs{T_1} + \abs{T_2} - \ed(T_1, T_2) = 6 + 6 - 3 = 9$, and the mapping indeed has weight 9 since the first pair contributes $1$ to the weight and the remaining 4 pairs contribute $2$ to the weight.

\begin{figure}[H]
    \centering 
    \begin{tikzpicture}[
scale=0.5,
every node/.style={solid},
edge from parent/.style={draw,solid}, 
level/.style={sibling distance=20mm},
level distance=30mm
]
\node [circle,draw,label=above:{$T_1$}] at (0, 0) {$a_1$}
    [child anchor=north]
    child {node [circle,draw] {$b_2$} edge from parent[solid]
    }
    child {node [circle,draw] {$c$} edge from parent[solid]
        child {node [circle,draw] {$b_3$} edge from parent[solid]
        }
        child {node [circle,draw] {$d_4$} edge from parent[solid]
        };
    }
    child {node [circle,draw] {$d_5$} edge from parent[solid]
    };
    
\node [circle,draw,label=above:{$T_2$}] at (8, 0) {$e_1$}
    [child anchor=north]
    child {node [circle,draw] {$b_2$} edge from parent[solid]
    }
    child {node [circle,draw] {$b_3$} edge from parent[solid]
    }
    child {node [circle,draw] {$b$} edge from parent[solid]
        child {node [circle,draw] {$d_4$} edge from parent[solid]
        }
        child {node [circle,draw] {$d_5$} edge from parent[solid]
        };
    };
\end{tikzpicture}
\caption{The mapping of maximum weight between $T_1$ and $T_2$} \label{fig:mapping}
\end{figure}
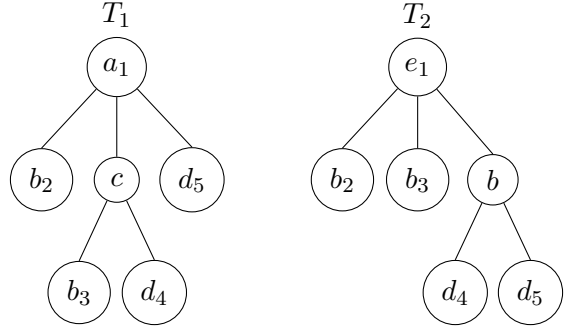

\subsection{Matrix-related definitions and results}
\begin{definition} [Similarity matrix] \label{def:similaritymatrix}
For two forests $F_1$ and $F_2$, let the \emph{similarity matrix} $S(F_1, F_2) = s_{i, j}(F_1)$ be a $(2\abs{F_2} + 1) \times (2\abs{F_2} + 1)$ matrix where
\begin{align*}
    s_{i, j}(F_1) &=
    \begin{cases}
        \similarity(F_1, F_2[i, j))                  & \text{if } i \le j \\
        -\infty                              & \text{if } i > j 
    \end{cases}.
\end{align*}
\end{definition}

For two compatible matrices $A$ and $B$, we use $A \star B$ to denote the \emph{max-plus product} of $A$ and $B$, which is a matrix $C = c_{ij}$ where $c_{ij} = \max_k{\{a_{ik} + b_{kj}\}}$.

An $n \times m$ matrix $A = a_{ij}$ is called \emph{row-monotone} if $a_{i, j} \le a_{i, j + 1}$ for all $i, j$, and \emph{column-monotone} if $a_{i + 1, j} \le a_{i, j}$ for all $i, j$. 

An $n \times n$ matrix $A$ is called \emph{finite-upper-triangular} if the entries below the main diagonal are $-\infty$ and the entries elsewhere are finite, and is called $W$\emph{-bounded-upper-triangular} if $A$ is finite-upper-triangular and the non-$(-\infty)$ entries of $A$ are integers between $0$ and $W$.

\begin{definition}[$W$-bounded-difference] \label{def:bd}
    An $n \times m$ matrix $A = a_{ij}$ is a \emph{$W$-bounded-difference} matrix if for all $i, j$, we have
    \begin{align*}
        \abs{a_{i, j} - a_{i - 1, j}} \le W, \\
        \abs{a_{i, j} - a_{i, j + 1}} \le W. \\
    \end{align*}
    When $W = O(1)$, we say matrix $A$ is a \emph{bounded-difference} matrix.
\end{definition}
A finite-upper-triangular $n \times n$ matrix $M$ is a \emph{finite-upper-triangular-$W$-bounded-difference} matrix if the property in Definition \ref{def:bd} holds for all $i \le j$.

The folloing result by Bringmann et. al. \cite{Bringmann16} is important for our truly sub-cubic running time:
\begin{theorem} [Theorem 1 of \cite{Bringmann16}] \label{theorem:bringmann}
    There is an $O(n ^ {2.8244})$ time randomized algorithm and an $O(n ^ {2.8603})$ time deterministic algorithm that computes the min-plus product of any two $n \times n$ bounded-difference matrices.
\end{theorem}
By negating all the entries in the matrices, their algorithm can also apply to max-plus products.

\subsection{A computation model for row-monotone, column-monotone matrices} \label{sec:model}
To speed up matrix manipulations, we store and modify our matrices in an implicit way. We now define a computation model for the matrices involved in our algorithm.

We consider the following \emph{range operations} for $n \times m$ row-monotone, column-monotone matrices. Given an $n \times m$ matrix $A = a_{ij}$, we can
\begin{itemize}
    \item Produce a new $n \times m$ matrix $B = [-\infty]_{n, m}$ (i.e. entries of $B$ are all $-\infty$),
    \item Produce a new matrix $B = A$, or
    \item Given $i ^ {\prime}, j ^ {\prime}$ and $x \in \mathbb{N} \cup \{-\infty\}$, produce a new matrix $B = b_{ij}$ such that
    \begin{align*}
        b_{ij} = 
        \begin{cases}
            \max(a_{ij}, x)                     & \text{if } i \le i ^ {\prime} \text{ and } j ^ {\prime} \le j \\
            a_{ij}                              & \text{otherwise}  
        \end{cases}
    \end{align*}
    and denote such matrix as $B = \rangemax(A, i ^ {\prime}, j ^ {\prime}, x)$.
\end{itemize}
We also consider the following \emph{range queries} on an $n \times m$ row-monotone, column-monotone matrix $A$:
\begin{itemize}
    \item Given $i, j$, query $A_{ij}$,
    \item Given $i, x$, query $\mincol(A, i, x) = \min\{j \mid A_{ij} \ge x\}$ or any index in $[1, m]$ if such $j$ does not exist, or
    \item Given $j, x$, query $\maxrow(A, j, x) = \max\{i \mid A_{ij} \ge x\}$ or any index in $[1, n]$ if such $i$ does not exist. 
\end{itemize}
There are well-known data structures (e.g. \emph{persistent 2D segment trees}) that can perform the mentioned range operations and range queries in $\tilde O(1)$ time assuming that all matrices involved are created using our model.

\section{Review of previous algorithms} \label{sec:review}

Per convention, given a forest $F$, let $\ell_F = \Root(L_F)$ and $r_F = \Root(R_F)$. Let $L'_F$ denote $F - R_F$, $R'_F$ denote $F - L_F$, $R_F^{\circ}$ denote $R_F - r_F$ and $L_F^{\circ}$ denote $L_F-\ell_F$.

\paragraph{Zhang and Shasha's algorithm and its decendants.} The current mainstream algorithms for the tree edit distance problem are based on Zhang and Shasha's first $O(n ^ 4)$ algorithm in 1989, assuming that both trees are of size $n$. It computes $\ed(F_1,F_2)$ for two forests $F_1$ and $F_2$ recursively as follows \cite{ZhangS89}:
\begin{itemize}
    \item If either of $F_1,F_2$ is empty, we have
    \begin{equation}
        \ed(F_1, \emptyset) = \abs{F_1}, \ed(\emptyset, F_2) = \abs{F_2}.
    \end{equation}
    \item Otherwise, we recurse with
    \begin{equation} \label{eq:recurse}
        \ed(F_1, F_2)  = \min
        \begin{cases}
            \ed(F_1-r_{F_1},F_2)+1 \\
            \ed(F_1, F_2-r_{F_2})+1 \\
            \ed(R^\circ_{F_1},R^\circ_{F_2}) + \ed(L'_{F_1}, L'_{F_2}) + \delta(r_{F_1},r_{F_2})
        \end{cases}.
    \end{equation}
\end{itemize}
    
Note that our choice of matching from the right side in (\ref{eq:recurse}) is arbitrary. Klein \cite{Klein98} improves the algorithm to $O(n ^ 3\log n)$ time by matching from the side determined by a rule: if $\size(L_{F_1}) > \size(R_{F_1})$, then we still recurse with (\ref{eq:recurse}), but if $\size(L_{F_1}) \le \size(R_{F_1})$, we instead recurse with
\begin{equation}
    \ed(F_1, F_2)  = \min
    \begin{cases}
        \ed(F_1-\ell_{F_1},F_2)+1 \\
        \ed(F_1, F_2-\ell_{F_2})+1 \\
        \ed(L^\circ_{F_1},L^\circ_{F_2}) + \ed(R'_{F_1}, R'_{F_2}) + \delta(\ell_{F_1},\ell_{F_2})
    \end{cases}.
\end{equation}
        
Since an improved strategy in the direction of recursion gives us an improvement in running time, it is natural to ask whether we can further improve this running time by further improving our strategy. We note that in Klein's algorithm, we always make our decision based on the sizes of the trees in $F_1$. In \cite{demaine2007}, a strategy with the idea of switching the roles of $F_1$ and $F_2$ when $\abs{F_2} > \abs{F_1}$ was introduced and proved to be \textit{optimal}, but it only improved the algorithm to $O(n ^ 3)$. This means that no improvement of Zhang and Shasha's algorithm in this way can break the cubic barrier.

\paragraph{Limitation of Zhang and Shasha's dynamic programming scheme.} In fact, it is no surprising that improvement of Zhang and Shasha's dynamic programming scheme cannot give us a sub-cubic running time. Recall that the weighted case of tree edit distance is unlikely to be solvable in truly sub-cubic time \cite{apsphard2020}, so in order to break the cubic barrier in the unweighted case, we need to find an algorithm that has a running time specific to the unweighted case. We note that the transition types in Zhang and Shasha's dynamic programming scheme are rather limited, making it unlikely that some variant of the algorithm could have different running times between the unweighted case and the weighted case. To illustrate our point, let $N = 2\abs{F_2} + 1$. We note that
\begin{align*}
    (L ^ \circ_{F_2}, R'_{F_2}) &= (F_2[1, \Right(\ell_{F_2}) - 1), F_2[\Right(\ell_{F_2}), N)), \\
    (L'_{F_2}, R ^ \circ_{F_2}) &= (F_2[1, \Left(r_{F_2})), F_2[\Left(r_{F_2}) + 1, N)), \\
    F_2 - \ell_{F_2} &= F_2[2, N), \\
    F_2 - r_{F_2} &= F_2[1, N - 1).
\end{align*}
We can see that the transitions only involve $F_2[1, k)$ for $k \in \{\Left(\ell_{F_2}), \Right(\ell_{F_2}) - 1, N - 1\}$, and $F_2[k, N)$ for $k \in \{\Right(r_{F_2}), \Left(r_{F_2}) + 1, 2\}$. Thus only a constant amount of values of $k$ are involved. We now introduce an algorithm that uses more types of transitions.

\paragraph{Chen's algorithm and its relation to our algorithm.} Chen's algorithm from 2001 \cite{chen01} is as follows:
\begin{itemize}
    \item If both $F_1,F_2$ contain only one tree, we recurse with
    \begin{equation}
        \ed(F_1, F_2) = \min
        \begin{cases}
            \ed(F_1-\Root(F_1),F_2) + 1 \\
            \ed(F_1, F_2 - \Root(F_2)) + 1 \\
            \ed(F_1 - \Root(F_1), F_2 - \Root(F_2))+ \delta(\Root(F_1), \Root(F_2))
        \end{cases}.
    \end{equation}
    \item If $F_1$ is a forest and $F_2$ is a tree, we recurse with
    \begin{equation}
        \ed(F_1, F_2) = \min
        \begin{cases}
            \ed(F_1,F_2-\Root(F_2)) + 1 \\
            \ed(R_{F_1}, F_2) + \abs{F_1} - \abs{R_{F_1}} \\
            \ed(F_1 - R_{F_1}, F_2) + \abs{R_{F_1}}
        \end{cases}.
    \end{equation}
    \item If $F_1$ is a tree and $F_2$ is a forest, we recurse with
    \begin{equation}
        \ed(F_1, F_2) = \min
        \begin{cases}
            \ed(F_1 - \Root(F_1), F_2) + 1 \\
            \operatorname*{min}\limits_{\textrm{tree } T \in F_2} {\{\ed(F_1, T) + \abs{F_2} - \abs{T}\}}
        \end{cases}.
    \end{equation}
    \item If both $F_1$ and $F_2$ are forests, let $\Delta(F_2, x) = \abs{F_2} - \abs{F_2[1, x)} - \abs{F_2[x, N)}$, and we recurse with
    \begin{equation} \label{eq:leaves}
        \ed(F_1, F_2)  = \min
        \begin{cases}
            \ed(R_{F_1}, F_2) + \abs{F_1} - \abs{R_{F_1}} \\
            \ed(F_1 - R_{F_1}, F_2) + \abs{R_{F_1}} \\
            \operatorname*{min}\limits_{l \in \textrm{leaves}(F_2)} {\{\ed(F_1 - R_{F_1}, F_2[1, \Right(l) + 1)) + \ed(R_{F_1}, F_2[\Right(l) + 1, N)) + \Delta(F_2, \Right(l) + 1)\}} 
        \end{cases}.
    \end{equation}
\end{itemize}
Chen's algorithm is notably efficient when the number of leaves in the input forests is small, but has a running time of $O(n ^ 4)$ in the general case. Chen himself erroneously claimed a running time of $O(n ^ {3.5})$ in \cite{chen01}. His mistake was only pointed out in a recent work by Schwarz, Pawlik and Augsten \cite{schwarz2017}. Despite its slow running time, Chen's algorithm has more general transition types compared to Zhang and Shasha's algorithm and its descendants --- $F_2[1, k)$ and $F_2[k, N)$ are involved for much more possible values of $k$ in (\ref{eq:leaves}).

We now show how Chen's dynamic programming scheme relates to ours, which will be fully introduced in Section \ref{sec:dp}. The significance of considering the equivalent maximization problem on similarity is that if we write out the equivalent transition of (\ref{eq:leaves}) for similarity, we get
\begin{equation} \label{eq:leavessim}
    \similarity(F_1, F_2)  = \max
    \begin{cases}
        \similarity(R_{F_1}, F_2) \\
        \similarity(F_1 - R_{F_1}, F_2) \\
        \operatorname*{max}\limits_{l \in \textrm{leaves}(F_2)} {\{\similarity(F_1 - R_{F_1}, F_2[1, \Right(l) + 1)) + \similarity(R_{F_1}, F_2[\Right(l) + 1, N))\}}
    \end{cases},
\end{equation}
where no extra terms involving sizes of subforests are present. We can further generalize this into
\begin{equation} \label{eq:leavesmul}
    \similarity(F_1, F_2)  = \operatorname*{max}\limits_{1 \le k \le N}{\{\similarity(F_1 - R_{F_1}, F_2[1, k))) + \similarity(R_{F_1}, F_2[k, N))\}}.
\end{equation}
The transition can now be written as max-plus product: $S(F_1, F_2) = S(F_1 - R_{F_1}, F_2) \star S(R_{F_1}, F_2)$. The price we pay is that the new dynamic programming scheme is no longer as efficient when the number of leaves in $F_2$ is bounded, but this is an affordable loss since we consider the general case where there is no limit on the number of leaves.

Finally, we want to point out that although our algorithm is the first truly sub-cubic algorithm for the unweighted tree edit distance problem, it is not the first ever to relate tree edit distance with matrix multiplication, since Chen himself was able to reduce his dynamic programming scheme to min-plus product using a different approach \cite{chen01, schwarz2017}.

\section{Our algorithm} 
\label{sec:algorithm}

In this section we introduce our algorithm in detail. In Section \ref{sec:dp} we introduce the dynamic programming scheme we start with and a novel cubic algorithm based on this scheme, and motivate the decomposition scheme that we will use to break the cubic barrier. Section \ref{sec:decomposition} introduces our main algorithm based on that decomposition scheme, which involves two different types of transitions. Section \ref{sec:transition} introduces how the type II transitions between two synchronous subforests with small difference in sizes can be done. Finally, Section \ref{sec:mul} introduces how the type I transitions which require more efficient computation of max-plus products can be done via a reduction to max-plus product between bounded-difference matrices.

We will assume that the input forests are two trees $T_1$ and $T_2$. It is easy to extend our algorithm to generic forests. We also assume that $\abs{T_1} \ge \abs{T_2}$. We will use the shorthand $S(F) = S(F, T_2)$ for any forest $F$ since the second argument will always be $T_2$ in our algorithm.

\subsection{A novel cubic algorithm} \label{sec:dp}

\subsubsection{Our dynamic programming scheme}  Given a forest $F$, we use the following dynamic programming scheme to compute $S(F)$: 
\begin{itemize}
    \item Basic case: $S(\emptyset)$ is a $(2\abs{T_2} + 1) \times (2\abs{T_2} + 1)$ finite-upper-triangular matrix whose entries on or above the main diagonal are all zero. 
    \item If $F$ contains one tree, let $u = \Root(F)$ and we recurse with
    \begin{align} \label{eq:dp1}
        s_{i, j}(F)       = \max
        \begin{cases}
            s_{i, j}(F - u) \\
            \operatorname*{max}\limits_{v \in T_2[i, j)}{\{s_{\Left(v) + 1, \Right(v) - 1}(F - u) + \eta(u, v)\}}
        \end{cases}.
    \end{align}
    \item Otherwise, as justified at the end of Section \ref{sec:review} (Equations \ref{eq:leavessim} and \ref{eq:leavesmul}), we recurse with
    \begin{align} \label{eq:dp2}
        S(F) = S(F - R_F) \star S(R_F).
    \end{align}
\end{itemize}

One can see that to compute $S(T_1)$, our dynamic programming scheme recursively computes exactly $S(\sub(u))$ and $S(\sub(u, [1, k]))$ for all $u \in T_1, 2 \le k \le \childcount(u)$, which gives us $O(\abs{T_1})$ similarity matrices in total.

Since we need to compute max-plus products between $(2\abs{T_2} + 1) \times (2\abs{T_2} + 1)$ matrices $O(\abs{T_1})$ times, the total running time is $\tilde O(\abs{T_1}\abs{T_2} ^ 3)$. To show that our new dynamic programming approach is promising, we now show how to improve the running time to cubic by exploiting the properties of the matrices using simple combinatorial methods.

\subsubsection{Properties of similarity matrices} We note that for any forest $F$, $S(F)$ is row-monotone and column-monotone: for $i ^ {\prime} \le i \le j \le j ^ {\prime}$, $T_2[i, j) \subset T_2[i ^ {\prime}, j ^ {\prime})$ and a mapping from $F$ to $T_2[i, j)$ is also a mapping from $T_2[i ^ {\prime}, j ^ {\prime})$. We also note that $S(F)$ is $2\min(\abs{F}, \abs{T_2})$-bounded-upper-triangular: there are at most $\min(\abs{F}, \abs{T_2})$ mapped pairs in a mapping between $F$ and $T_2$ and each pair only contributes at most 2 to the answer. Finally, we show the following:
\begin{lemma} \label{lemma:bd}
    For any forest $F$, $S(F)$ is a finite-upper-triangular-2-bounded-difference matrix.
\end{lemma}
\begin{proof}
    We show $\similarity(F, T_2[i, j + 1)) \le \similarity(F, T_2[i, j)) + 2$. The other direction is similar. If $T_2[i, j + 1) = T_2[i, j)$ then $\similarity(F, T_2[i, j + 1)) = \similarity(F, T_2[i, j))$. Otherwise let $u$ be the unique node in $T_2[i, j + 1) \backslash T_2[i, j)$. If $u$ is not mapped in the mapping that maximizes $\similarity(F, T_2[i, j + 1))$, then the mapping is also a valid mapping between $F$ and $T_2[i, j)$ and we have $\similarity(F, T_2[i, j + 1)) \le \similarity(F, T_2[i, j))$. Otherwise if $u$ maps to $v \in T_2[i, j)$, remove $u$ and $v$ from the mapping and we get a valid mapping between $F$ and $T_2[i, j)$. Therefore we have $\similarity(F, T_2[i, j + 1)) - \eta(u, v) \le \similarity(F, T_2[i, j))$. Since $\eta(u, v) \le 2$, we have $\similarity(F, T_2[i, j + 1)) \le \similarity(F, T_2[i, j)) + 2$.
\end{proof}

\subsubsection{Optimization to cubic} 
We now show how to compute the similarity matrices involved in our dynamic programming scheme in cubic time, relying on the fact that $S(F)$ is row-monotone, column-monotone and $2\abs{F}$-bounded-upper-triangular. Our main sub-cubic algorithm will also rely on the finite-upper-triangular-bounded-difference property, where the cubic algorithm here is used as a sub-routine for obtaining similarity matrices for small sub-forests of $T_1$. 

We first prove the following theorem, where we are using the range operation/query model from Section \ref{sec:model}.
\begin{theorem} \label{theorem:cubic}
    For a forest $F$, $S(F)$ can be computed in $\tilde O(\abs{F} ^ 2\abs{T_2})$ time.
\end{theorem}
Note the apparently nonsensical running time: if $\abs{F} = o(\abs{T_2} ^ {0.5})$ , then $S(F)$ can be computed in $o(\abs{T_2} ^ 2)$ time, smaller than the number of entries in $S(F)$. This is because we are using the model in Section \ref{sec:model}, and the running time here means that we can answer all the defined range queries on $S(F)$ after some $o(\abs{T_2} ^ 2)$ time pre-processing.

In order to prove Theorem \ref{theorem:cubic}, we will need to show the following:
\begin{lemma} \label{lemma:mulbasic}
    Let $A, B$ be $n \times n$ row-monotone, column-monotone matrices. If $A$ is $m_A$-bounded-upper-triangular and $B$ is $m_B$-bounded-upper-triangular, then $C = A \star B$ can be computed in $\tilde O(m_Am_Bn)$ time.
\end{lemma}
\begin{proof} [Proof of Theorem \ref{theorem:cubic} from Lemma \ref{lemma:mulbasic}]
    Consider applying Lemma \ref{lemma:mulbasic} to our dynamic programming scheme. Every time we multiply the similarity matrices of two subforests $F_1$ and $F_2$, we contribute $\tilde O(\abs{F_1}\abs{F_2}\abs{T_2})$ to the total running time. We know the sum of $\abs{F_1}\abs{F_2}$ across all multiplications is $O(\abs{F} ^ 2)$ from a classic argument: $\abs{F_1}\abs{F_2}$ is the number of pairs of nodes $(x, y)$ where $x \in F_1$ and $y \in F_2$ and each $(x, y)$ pair will only be counted once. Therefore the total running time is $\tilde O(\abs{F} ^ 2\abs{T_2})$. 
\end{proof}
This proof also implies the following corollary, which will be used again in our main truly sub-cubic algorithm:
\begin{cor} \label{cor:cubic}
    Let $\{F_1, F_2, \cdots, F_k\}$ be a sequence of forests. For $1 \le l \le k$, let $G_l = F_1 + F_2 + \cdots + F_l$ be the $l$-th prefix sum. Then the sequence of similarity matrices of prefix sums $\{S(G_1), S(G_2) \cdots S(G_k)\}$ can be computed in $\tilde O(\abs{G_k} ^ 2\abs{T_2})$ time. The same result holds for suffix sums.
\end{cor}

To prove Lemma \ref{lemma:mulbasic}, Algorithm \ref{algorithm:mulbasic} computes $C = A \star B$ in $\tilde O(m_Am_Bn)$ time.
\begin{algorithm} [H]
    \caption{Computation of $C = A \star B$ in Lemma \ref{lemma:mulbasic}} \label{algorithm:mulbasic}
    \begin{algorithmic}[1]
        \Procedure{MUL1}{$A, B$}
            \State $C \gets [-\infty]_{n, n}$
            \For {$j \in [1, n]$}
                \For {$x \in [0, m_B]$}
                    \For {$y \in [0, m_A]$}
                        \State $k \gets \maxrow(B, j, x)$
                        \State $i \gets \maxrow(A, k, y)$
                        \State $C \gets \rangemax(C, i, j, A_{i, k} + B_{k, j})$ \label{line:covertriple}
                    \EndFor
                \EndFor
            \EndFor
        \EndProcedure
    \end{algorithmic}
\end{algorithm}
To show its correctness, a triple $(i, k, j)$ is \textit{useful} if $A_{i, k} + B_{k, j} \ge C_{i, j}$ and $\max(A_{i + 1, k} + B_{k, j}, A_{i, k + 1} + B_{k + 1, j}) < C_{i, j}$. A triple $(i, k, j)$ is \textit{covered} by the algorithm if $\rangemax(C, i, j, A_{i, k} + B_{k, j})$ has been produced on line \ref{line:covertriple}. It suffices to show that all useful triples are covered. We can see that a triple $(i, k, j)$ cannot be useful when
\begin{itemize}
    \item $A_{i, k} = A_{i + 1, k}$, since if $A_{i, k} + B_{k, j} \ge C_{i, j}$, then $A_{i + 1, k} + B_{k, j} = A_{i, k} + B_{k, j} \ge C_{i, j}$, or
    \item $B_{k, j} = B_{k + 1, j}$, since if $A_{i, k} + B_{k, j} \ge C_{i, j}$, then $A_{i, k + 1} + B_{k + 1, j} \ge A_{i, k} + B_{k + 1, j} = A_{i, k} + B_{k, j} \ge C_{i, j}$.
\end{itemize}
Therefore, $A_{i, k}$ must be the last occurrence of its value on column $k$ of $A$ and $B_{k, j}$ must be the last occurrence of its value on column $j$ of $B$. We can now see that the algorithm indeed covers all useful triples.

Note that in our dynamic programming scheme we also need to deal with the case of Equation \ref{eq:dp1}. This can easily be done in $\tilde O(\abs{T_2})$ time by first initializing $S(F)$ with $S(F - u)$ and then enumerating the $O(\abs{T_2})$ nodes $v \in T_2[i, j)$ and replacing $S(F)$ with $\rangemax(S(F), l(v), r(v), s_{\Left(v) + 1, \Right(v) - 1}(F - u) + \eta(u, v))$.

Finally, it seems that the logarithmic factors in our cubic algorithm can be removed by using some more time-efficient way to store the matrices (e.g. by maintaining two 2D tables for a similarity matrix $A$ storing the values of $\mincol(A, i, x)$ and $\maxrow(A, j, x)$). We do not give the full details here since we focus on polynomial improvements.

\paragraph{Breaking the cubic barrier?} Unfortunately, the running time of our dynamic programming scheme cannot be improved to truly sub-cubic in a straightforward way. Consider the amount of information necessary to represent $S(\sub(u))$ for some $u \in T_1$. Intuitively, for each row $i$, we need the value of $\mincol(S(\sub(u)), x)$ for every $x \in [0, 2\min(\abs{T_2}, \abs{\sub(u)})]$, which gives $O(\abs{T_2}\min(\abs{T_2}, \abs{\sub(u)}))$ different $(i, x)$ pairs. Since the sum of $\min(\abs{T_2}, \abs{\sub(u)})$ across all $u \in T_1$ is $O(\abs{T_1}\abs{T_2})$, the total amount of information needed to represent the similarity matrices for subtrees of every node in $T_1$ is already $O(\abs{T_1}\abs{T_2} ^ 2)$. It seems impossible for us to be able to find a more efficient way to represent all these similarity matrices. Therefore, in order to break the cubic barrier, we cannot compute $S(\sub(u))$ one by one for every $u \in T_1$. In our main algorithm, we use a \emph{decomposition scheme} to decompose the transitions into blocks in order to skip some of the nodes in $T_1$.

\subsection{Main algorithm} \label{sec:decomposition}

\subsubsection{Overview}
Our main algorithm is based on the following decomposition scheme: for a block size $\Delta$, we decompose the computation of $S(T_1)$ into $O(\abs{T_1} / \Delta)$ transitions between synchronous subforests (Definition \ref{def:syn}) of the following two types:
\begin{itemize}
    \item Type I: transition from two synchronous subforests $F_1$ and $F_2$ to $F_1 + F_2$, where $F_1 + F_2$ is also a synchronous subforest and both $\abs{F_1}$ and $\abs{F_2}$ are no less than $\Delta$.
    \item Type II: transition from synchronous forest $F_1$ to synchronous forest $F_2$ such that $F_1 \subset F_2$ and $\abs{F_2} - \abs{F_1} = O(\Delta)$.
\end{itemize}

For the type I transitions, in Section \ref{sec:mul} we will show the following theorem.
\begin{theorem} \label{theorem:mul}
    Let $A, B$ be row-monotone, column-monotone and finite-upper-triangular-bounded-difference $n \times n$ matrices whose entries on the main diagonals are zero. If $A$ is $m$-bounded-upper-triangular, then $C = A \star B$ can be computed in $\MUL(m, n) = \tilde O(m ^ {0.9038}n ^ 2)$ randomized and $\MUL(m, n) = \tilde O(m ^ {0.9250}n ^ 2)$ deterministic time. 
\end{theorem}
Theorem \ref{theorem:mul} implies that $S(F_1) \star S(F_2)$ can be computed in $\MUL(\min(\abs{F_1}, \abs{F_2}, \abs{T_2}), \abs{T_2})$ time.

For the type II transitions, in Section \ref{sec:transition} we will show the following:
\begin{theorem} \label{theorem:transition}
    Let $F$ be a forest and let $F ^ {\prime}$ be a synchronous subforest of $F$. If $S(F ^ {\prime})$ is known, then $S(F)$ can be computed in $\tilde O(\MUL(\abs{F} - \abs{F ^ {\prime}}, \abs{T_2}) + \abs{T_2}(\abs{F} - \abs{F ^ {\prime}}) ^ 3)$ time, where $\MUL(\abs{F} - \abs{F ^ {\prime}}, \abs{T_2})$ refers to the running time in Theorem \ref{theorem:mul}.
\end{theorem}

In Section \ref{sec:implementation} we give our implementation of the decomposition scheme, and in Section \ref{sec:runtimeanalysis} we analyze its running time using Theorem \ref{theorem:mul} and Theorem \ref{theorem:transition}.

\subsubsection{Implementation} \label{sec:implementation}
Algorithm \ref{algorithm:decomposition} implements the decomposition scheme. For a synchronous subforest $F$, if $F$ contains more than one trees and both $\abs{L_F}$ and $\abs{R_F}$ are no less than $\Delta$, we compute $S(F)$ using a type I transition from $L_F$ and $F - L_F$. Otherwise, we find a synchronous subforest $F ^ {\prime}$ of $F$ such that $\abs{F} - \abs{F ^ {\prime}}$ is $O(\Delta)$ and use a type II transition from $F ^ {\prime}$ to $F$. To do this, if $\abs{F} \le 3\Delta$ we let $F ^ {\prime} = \emptyset$. Otherwise we let $F ^ {\prime} = F$, keep removing some part from $F ^ {\prime}$ and stop when the next removal will result in $\abs{F} - \abs{F ^ {\prime}}$ being greater than $2\Delta$.

We now show that the total amount of transitions is indeed $O(\abs{T_1} / \Delta)$. Note that each time we have a type I transition, we merge two subforests both of sizes no less than $\Delta$, so the total number of such transitions is $O(\abs{T_1} / \Delta)$. For the type II transitions, we further divide them into two cases:
\begin{itemize}
    \item (The first case) $F ^ {\prime}$ contains less than two trees or one of $\abs{L_{F ^ {\prime}}}$ and $\abs{R_{F ^ {\prime}}}$ is less than $\Delta$, or
    \item (The second case) $F ^ {\prime}$ contains at least two trees and both $\abs{L_{F ^ {\prime}}}$ and $\abs{R_{F ^ {\prime}}}$ are no less than $\Delta$.
\end{itemize}
For the first case, we have $\abs{F} - \abs{F ^ {\prime}} > \Delta$ since otherwise we would not have stopped the removal process. Note that $F \backslash {F ^ {\prime}}$ is disjoint across different type II transitions. Therefore there are no more than $\abs{T_1} / (\Delta + 1) = O(\abs{T_1} / \Delta)$ type II transitions of the first case. For the second case, we can see that the transition we will use to compute $S(F ^ {\prime})$ in the next recursive call will be type I. Therefore, the total number of type II transitions of the second case is bounded by the total number of type I transitions, which as we have already shown is $O(\abs{T_1} / \Delta)$.

\subsubsection{Total running time} \label{sec:runtimeanalysis}
We analyze the total running time of our algorithm for the randomized case only since the analysis for the deterministic case is nearly identical. 

The total running time for the type II transitions is $\tilde O((\abs{T_1} / \Delta)(\MUL(\Delta, \abs{T_2}) + \abs{T_2}\Delta ^ 3))$, which equals $$\tilde O(\abs{T_1}\abs{T_2} ^ 2 / \Delta ^ {0.0952} + \abs{T_1}\abs{T_2}\Delta ^ 2).$$

We relate the running time for the type I transitions to a value that is easier to analyze. Let $X$ be the total running time for the type I transitions, and let $Y$ be what the total running time for the type I transitions would be if $\MUL(\min(\abs{F_1}, \abs{F_2}, \abs{T_2}), \abs{T_2})$ equaled $O(\min(\abs{F_1}, \abs{F_2}, \abs{T_2})\abs{T_2} ^ 2)$ instead of $\tilde O(\min(\abs{F_1}, \abs{F_2}, \abs{T_2}) ^ {0.9038}\abs{T_2} ^ 2)$. Since $\min(\abs{F_1}, \abs{F_2}, \abs{T_2}) \ge \Delta$, $X = \tilde O(Y / \Delta ^ {0.0952})$. 
\begin{algorithm} [H]
    \caption{Computation of $S(F)$ by decomposition} \label{algorithm:decomposition}
    \begin{algorithmic}[1]
        \Procedure{COMPUTE}{$F$}
            \If {$L_F \ne R_F$ \AND $\abs{L_F} \ge \Delta$ \AND $\abs{R_F} \ge \Delta$} \Comment{Type I}
                \State COMPUTE($L_F$)
                \State COMPUTE($F - L_F$)
                \State Compute $S(F) = S(L_F) \star S(F - L_F)$ using Theorem \ref{theorem:mul}
            \Else \Comment{Type II}
                \If {$\abs{F} \le 3\Delta$}
                    \State $F ^ {\prime} \gets \emptyset$
                \Else
                    \State $F ^ {\prime} \gets F$
                    \While {TRUE}
                        \State $F_{\textrm{NEXT}} \gets F ^ {\prime}$
                        \If {$F ^ {\prime}$ contains only one tree}
                            \State $F_{\textrm{NEXT}} \gets F ^ {\prime} - \Root(F ^ {\prime})$
                        \Else  
                            \If {$\abs{L_{F ^ {\prime}}} < \abs{R_{F ^ {\prime}}}$}
                                \State $F_{\textrm{NEXT}} \gets F ^ {\prime} - L_{F ^ {\prime}}$
                            \Else
                                \State $F_{\textrm{NEXT}} \gets F ^ {\prime} - R_{F ^ {\prime}}$
                            \EndIf
                        \EndIf
                        \If {$\abs{F} - \abs{F_{\textrm{NEXT}}} > 2\Delta$}
                            \State \Break
                        \EndIf
                        \State $F ^ {\prime} \gets F_{\textrm{NEXT}}$
                    \EndWhile
                    \State COMPUTE($F ^ {\prime}$)
                \EndIf
                \State Compute $S(F)$ from $S(F ^ {\prime})$ using Theorem \ref{theorem:transition}
            \EndIf
        \EndProcedure
    \end{algorithmic}
\end{algorithm}

We now obtain a bound on $Y$. If both $\abs{F_1}$ and $\abs{F_2}$ are greater than $\abs{T_2}$, the total number of such transitions is $O(\abs{T_1} / \abs{T_2})$, so the total running time is $O((\abs{T_1} / \abs{T_2}) \times \abs{T_2} ^ 3) = O(\abs{T_1}\abs{T_2} ^ 2)$. If $\min(\abs{F_1}, \abs{F_2}) \le \abs{T_2}$, we adapt the classic argument on small to large merging: when we are multiplying the similarity matrices for two subforests, each node in the smaller subforest contributes $O(\abs{T_2} ^ 2)$ to the total running time. Since the size of the smaller forest is no more than $\abs{T_2}$, the sum of the total number of nodes in the smaller forests across all multiplications is $O(\abs{T_1} \log \abs{T_2}) = \tilde O(\abs{T_1})$, and we have $Y = \tilde O(\abs{T_1}\abs{T_2} ^ 2)$. Therefore $X = \tilde O(\abs{T_1}\abs{T_2} ^ 2 / \Delta ^ {0.0952})$. 

The total running time for the two types combined is 
\begin{align*}
    \tilde O(\abs{T_1}\abs{T_2} ^ 2 / \Delta ^ {0.0952} + \abs{T_1}\abs{T_2}\Delta ^ 2)
\end{align*}
where the hidden sub-polynomial factors are not dependent on $\abs{T_1}$. By setting $\Delta \approx \abs{T_2} ^ {0.4773}$ we get the bound in Theorem \ref{theorem:main}.

\subsection{Transition between synchronous subforests} \label{sec:transition}

In this section we show that Theorem \ref{theorem:transition} is true, thereby verifying that the type II transitions can be done within the desired time bound. 

To get the similarity matrix for forest $F$ given the similarity matrix for $F ^ {\prime}$, we need to consider mapping the nodes in $F \backslash F ^ {\prime}$ to $T_2$. As shown in Figure \ref{fig:spine}, we consider the path between $\vroot(F)$ and $\vroot(F ^ {\prime})$ in forest $F$:
\begin{align*}
    \vroot(F) = u_0 \rightarrow u_1 \rightarrow u_2 \rightarrow \cdots \rightarrow u_k = \vroot(F ^ {\prime})
\end{align*}
where all nodes except $u_0$ is in $F$. For $1 \le i \le k$, let $l_i$ be the subforest consisting of subtrees of siblings of $u_i$ to the left of $u_i$ from left to right, and let $r_i$ be the subforest consisting of subtrees of siblings to the right of $u_i$ from left to right. We also let $l_{k + 1}$ be the subforest consisting of subtrees of children of $u_k$ to the left of $F ^ {\prime}$, and $r_{k + 1}$ be the subforest consisting of subtrees of children of $u_k$ to the right of $F ^ {\prime}$. For simplicity, in Figure \ref{fig:spine} subforests are drawn as if they were subtrees.

For $j \ge i$, let $l_{i, j}$ be the subforest $l_i + l_{i + 1} + \cdots + l_j$ and let $r_{i, j}$ be the subforest $r_j + r_{j - 1} + \cdots + r_i$. From Corollary \ref{cor:cubic} it is easy to see that $S(l_{i, j})$ and $S(r_{i, j})$ across all $1 \le i \le j \le k + 1$ can be computed in $O(\abs{T_2}(\abs{F} - \abs{F ^ {\prime}}) ^ 3)$ time.

If none of $u_1, u_2, \cdots, u_k$ is mapped, then the contribution to $S(F)$ will be $S(l_{1, k + 1}) \star S(F ^ {\prime}) \star S(r_{1, k + 1})$, which can be computed in $\MUL(\min(\abs{F}- \abs{F ^ {\prime}}, \abs{T_2}), \abs{T_2})$ time. For the case where at least one of $u_1, u_2, \cdots, u_k$ is mapped, we first define a restricted version of the similarity matrix of a \emph{tree} where the root must be mapped: 
\begin{definition} [Restricted similarity matrix]
For a tree $T$, the \emph{restricted similarity matrix} $\hat S(T) = {\hat s}_{i, j}(T)$ is a $(2\abs{T_2} + 1) \times (2\abs{T_2} + 1)$ matrix where
\begin{align*}
    {\hat s}_{i, j}(T) = 
    \begin{cases}
        \operatorname*{max}\limits_{v \in T_2[i, j)}{\{\similarity(T - \Root(T), \sub(v) - v) + \eta(\Root(T), v)\}}         & \text{if } i \le j \textrm{ and }T_2[i, j) \ne \emptyset \\
        -\infty                    & \text{if } i > j\textrm{ or }T_2[i, j) = \emptyset
    \end{cases}.
\end{align*}
\end{definition}

The transition from $S(F ^ {\prime})$ to $S(F)$ consists of three parts:
\begin{itemize}
    \item (Bottom) For the largest $y$ such that $u_y$ is mapped, the transition from $S(F ^ {\prime})$ to ${\hat S}(\sub(u_y))$,
    \item (Middle) For all $(x, y)$ pairs where $x < y$, the transition from ${\hat S}(\sub(u_y))$ to ${\hat S}(\sub(u_x))$, and
    \item (Top) For the smallest $x$ such that $u_x$ is mapped, the transition from ${\hat S}(\sub(u_x))$ to $S(F)$.
\end{itemize}

\subsubsection{Final transitions from restricted similarity matrices (top)}
We first show how to do the top transitions since they do not involve computation of the restricted similarity matrices. We show:
\begin{lemma} \label{lemma:transition1}
    If $\hat S(\sub(u_x))$ is known for all $1 \le x \le k$, $S(F)$ can be computed in $\tilde O(\MUL(\abs{F}- \abs{F ^ {\prime}}, \abs{T_2}, \abs{T_2}) + \abs{T_2}(\abs{F} - \abs{F ^ {\prime}}) ^ 3)$ time.
\end{lemma}

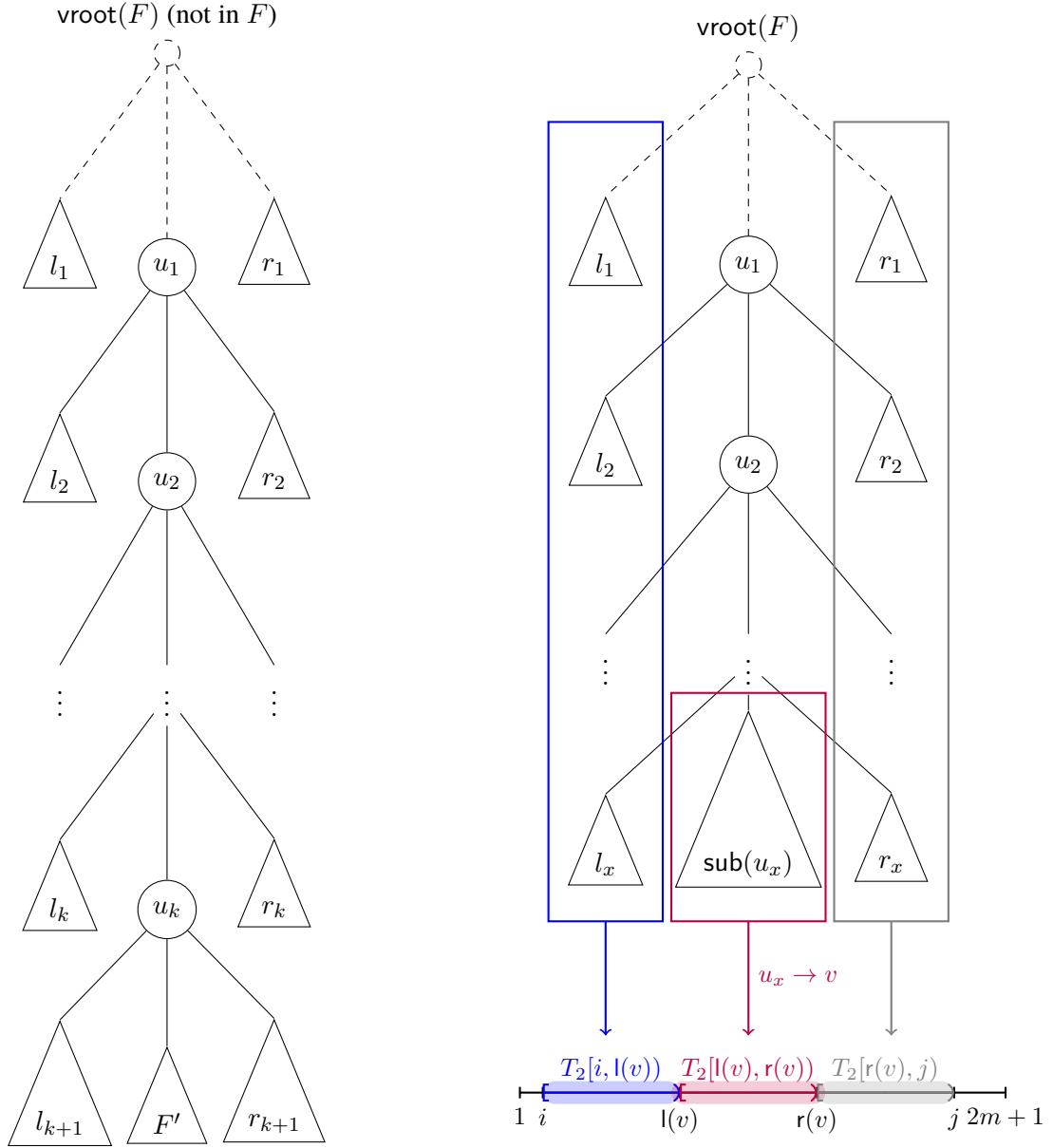
\begin{figure}[H]  
\centering 
    \begin{subfigure}[b]{0.5\linewidth}
        \begin{tikzpicture}[
every node/.style={solid},
edge from parent/.style={draw,solid}, 
level/.style={sibling distance=15mm},
level distance=30mm
]
\node [circle,dashed,draw,label=above:{$\vroot(F)$ (not in $F$)}] (vroot) {}
    [child anchor=north]
    child {node [triangle,draw] (l1) {$l_1$} edge from parent[dashed]
    }
    child {node [circle,draw] (u1) {$u_1$} edge from parent[dashed]
        child {node [triangle,draw] (l2) {$l_2$} 
        }
        child {node [circle,draw] (u2) {$u_2$}  edge from parent[solid]
            child {node {$\vdots$}  edge from parent[solid]
            }
            child {node {$\vdots$}  edge from parent[solid]
                child {node [triangle,draw] (lk) {$l_k$} 
                }
                child {node [circle,draw] (uk) {$u_k$}  edge from parent[solid]
                    child {node [triangle,draw] (lk1) {$l_{k + 1}$} 
                    }
                    child {node [triangle,draw] (fp) {$F ^ {\prime}$}
                    }
                    child {node [triangle,draw] (rk1) {$r_{k + 1}$}
                    };
                }
                child {node [triangle,draw] (rk) {$r_k$}
                };
            }
            child {node {$\vdots$}  edge from parent[solid]
            };
        }
        child {node [triangle,draw] (r2) {$r_2$}
        };
    }
    child {node [triangle,draw] (r1) {$r_1$} edge from parent[dashed]
    };
    \node[draw=none] at (-3.5,0) {};
\end{tikzpicture}
        \caption{The nodes $u_{1 \cdots k}$ and the subforests $l_{1 \cdots k + 1}$ and $r_{1 \cdots k + 1}$}
        \label{fig:spine}  
    \end{subfigure}
    \begin{subfigure}[b]{0.47\linewidth}
        \begin{tikzpicture}[
scale=0.8,
every node/.style={solid},
edge from parent/.style={draw,solid}, 
level/.style={sibling distance=25mm},
level distance=35mm
]
\node [circle,dashed,draw,label=above:{$\vroot(F)$}] (vroot) {}
    [child anchor=north]
    child {node [triangle,draw] (l1) {$l_1$} edge from parent[dashed]
    }
    child {node [circle,draw] (u1) {$u_1$} edge from parent[dashed]
        child {node [triangle,draw] (l2) {$l_2$} 
        }
        child {node [circle,draw] (u2) {$u_2$}  edge from parent[solid]
            child {node {$\vdots$}  edge from parent[solid]
            }
            child {node {$\vdots$}  edge from parent[solid]
                child {node [triangle,draw] (lx) {$l_x$} 
                }
                child {node [triangle,draw] (ux) {$\sub(u_x)$}  edge from parent[solid]
                }
                child {node [triangle,draw] (rx) {$r_x$}
                };
            }
            child {node {$\vdots$}  edge from parent[solid]
            };
        }
        child {node [triangle,draw] (r2) {$r_2$}
        };
    }
    child {node [triangle,draw] (r1) {$r_1$} edge from parent[dashed]
    };
\draw[thick] (-4,-18) -- (4.5,-18);
\foreach \x/\xtext in {-4/$1$,-3.6/$i$,-1.2/$\Left(v)$,1.2/$\Right(v)$,3.6/$j$, 4.5/$2m + 1$}
    \draw[thick] (\x,-17.9) -- (\x,-18.1) node[below] {\small \xtext};
    
\draw[[-, thick, blue] (-3.6,-18) -- (-2.4,-18)  node[above] {\small $T_2[i, \Left(v))$};
\draw[-), thick, blue] (-2.4,-18) -- (-1.2,-18);
\fill[opacity = 0.2, blue,rounded corners=1ex] (-3.6,-17.8) -- (-1.2, -17.8) -- (-1.2, -18.2) -- (-3.6, -18.2) -- cycle;

\draw[thick, blue] (-3.5, -1) -- (-3.5, -15) -- (-1.5, -15) -- (-1.5, -1) -- cycle;
\draw[->, thick, blue] (-2.5, -15) -- (-2.5, -17);

\draw[[-, thick, purple] (-1.2,-18) -- (0,-18)  node[above] {\small $T_2[\Left(v), \Right(v))$};
\draw[-), thick, purple] (0,-18) -- (1.2,-18);
\fill[opacity = 0.2, purple,rounded corners=1ex] (-1.35,-17.8) -- (1.35, -17.8) -- (1.35, -18.2) -- (-1.35, -18.2) -- cycle;

\draw[thick, purple] (-1.35, -11) -- (-1.35, -15) -- (1.35, -15) -- (1.35, -11) -- cycle;
\draw[->, thick, purple] (0, -15) -- (0, -17) node[midway, right] {\small $u_x \rightarrow v$};

\draw[[-, thick, gray] (1.2,-18) -- (2.4,-18)  node[above] {\small $T_2[\Right(v), j)$};
\draw[-), thick, gray] (2.4,-18) -- (3.6,-18);
\fill[opacity = 0.2, gray,rounded corners=1ex] (1.2,-17.8) -- (3.6, -17.8) -- (3.6, -18.2) -- (1.2, -18.2) -- cycle;

\draw[thick, gray] (1.5, -1) -- (1.5, -15) -- (3.5, -15) -- (3.5, -1) -- cycle;
\draw[->, thick, gray] (2.5, -15) -- (2.5, -17);

\end{tikzpicture}
        \caption{Transition from $\hat S(\sub(u_x))$ to $S(F)$ (top)}     \label{fig:top}  
    \end{subfigure}
\caption{Transition between synchronous subforests}
\end{figure}  
 
To compute $S(F)_{i, j} = \similarity(F, T_2[i, j))$, we can enumerate the smallest $x$ such that $u_x$ maps to some node $v \in T_2[i, j)$. Then as shown in Figure \ref{fig:top}, we have
\begin{itemize}
    \item Nodes in $l_{1, x}$ map to nodes in $T_2[i, \Left(v))$, contributing $s_{i, \Left(v)}(l_{1, x})$,
    \item Nodes in $\sub(u_x)$ map to nodes in $T_2[\Left(v), \Right(v))$, contributing ${\hat s}_{\Left(v), \Right(v)}(\sub(u_x))$, and
    \item Nodes in $r_{1, x}$ map to nodes in $T_2[\Right(v), j)$, contributing $s_{\Right(v), j}(r_{1, x})$.
\end{itemize}
The total contribution will be $s_{i, \Left(v)}(l_{1, x}) + {\hat s}_{\Left(v), \Right(v)}(\sub(u_x)) + s_{\Right(v), j}(r_{1, x})$. Note that ${\hat s}_{\Left(v), \Right(v)}(\sub(u_x))$ might actually correspond to a mapping where $u_x$ is mapped to some node other than $v$ in $T_2[\Left(v), \Right(v))$, but this still gives us a valid mapping overall and does not affect the correctness of our algorithm. There are still $O(\abs{T_2} ^ 3(\abs{F} - \abs{F ^ {\prime}}))$ quadruples $(i, j, x, v)$, but as shown in Algorithm \ref{algorithm:transition1}, we can cut the number of quadruples to $O(\abs{T_2}(\abs{F} - \abs{F ^ {\prime}}) ^ 3)$ using an idea similar to the one in Algorithm \ref{algorithm:mulbasic}.
\begin{algorithm}  [H]
    \caption{Computation of $S(F)$ from $\hat S(\sub(u_x))$} \label{algorithm:transition1}
    \begin{algorithmic}[1]
        \State $S(F) \gets S(l_{1, k + 1}) \star S(F ^ {\prime}) \star S(r_{1, k + 1})$ \Comment{No $u_x$ mapped}
        \For {$x \in [1, k]$} 
            \For {$v \in T_2$}
                \For {$y \in [0, 2(\abs{F} - \abs{F ^ {\prime}})]$}
                    \For {$z \in [0, 2(\abs{F} - \abs{F ^ {\prime}})]$}
                        \State $i \gets \maxrow(S(l_{1, x}), \Left(v), y)$
                        \State $j \gets \mincol(S(r_{1, x}), \Right(v), z)$
                        \State $S(F) \gets \rangemax(S(F), i, j, s_{i, \Left(v)}(l_{1, x}) + {\hat s}_{\Left(v), \Right(v)}(\sub(u_x)) + s_{\Right(v), j}(r_{1, x}))$
                    \EndFor
                \EndFor
            \EndFor
        \EndFor
    \end{algorithmic}
\end{algorithm}


\subsubsection{Computation of restricted similarity matrices (bottom/middle)}

We now finish the proof of Theorem \ref{theorem:transition} by showing
\begin{lemma} \label{lemma:transition2} 
    For any $1 \le x \le k$, if $\hat S(\sub(u_y))$ is known for all $y > x$, $\hat S(\sub(u_x))$ can be computed in $\tilde O(\abs{T_2}(\abs{F} - \abs{F ^ {\prime}}) ^ 2)$ time.
\end{lemma}
This enables us to compute $\hat S(\sub(u_x))$ for all $1 \le x \le k$ in decreasing order of $x$ with total time $O(\abs{T_2}(\abs{F} - \abs{F ^ {\prime}}) ^ 3)$.

To compute $\hat S(\sub(u_x))$ we need to consider two cases:
\begin{itemize}
    \item The first case (bottom): no $u_y (y > x)$ is mapped.
    \item The second case (middle): there exists some $y > x$ such that $u_y$ is mapped.
\end{itemize}

\paragraph{The first case (bottom).} This is the simpler case. We can enumerate the triple $(i, j, v)$ such that, as shown in Figure \ref{fig:bottom},
\begin{itemize}
    \item $u_x$ maps to $v$, contributing $\eta(u_x, v)$,
    \item Nodes in $l_{x + 1, k + 1}$ map to nodes in $T_2[\Left(v) + 1, i)$, contributing $s_{\Left(v) + 1, i}(l_{x + 1, k + 1})$,
    \item Nodes in $F ^ {\prime}$ map to nodes in $T_2[i, j)$, contributing $s_{i, j}(F ^ {\prime}) $, and
    \item Nodes in $r_{x + 1, k + 1}$ map to nodes in $T_2[j, \Right(v) - 1)$, contributing $s_{j, \Right(v) - 1}(r_{x + 1, k + 1})$.
\end{itemize} 
The total contribution is $\eta(u_x, v) + s_{\Left(v) + 1, i}(l_{x + 1, k + 1}) + s_{i, j}(F ^ {\prime}) + s_{j, \Right(v) - 1}(r_{x + 1, k + 1})$, and the contribution applies to ${\hat s}_{i ^ {\prime}, j ^ {\prime}}(\sub(u_x))$ for all $i ^ {\prime} \le \Left(v) < \Right(v) \le j ^ {\prime}$.

\paragraph{The second case (middle).} Without loss of generality we suppose no $x < y ^ {\prime} < y$ exists such that $u_{y ^ {\prime}}$ is mapped. We can enumerate the quadruple $(i, j, y, v)$ such that, as shown in Figure \ref{fig:mid},  
\begin{itemize}
    \item $u_x$ maps to $v$, contributing $\eta(u_x, v)$,
    \item Nodes in $l_{x + 1, y}$ map to nodes in $T_2[\Left(v) + 1, i)$, contributing $s_{\Left(v) + 1, i}(l_{x + 1, y})$,
    \item Nodes in $\sub(u_y)$ map to nodes in $T_2[i, j)$ and in particular $u_y$ is mapped to some $v ^ {\prime} \in T_2[i, j)$, contributing ${\hat s}_{i, j}(\sub(u_y))$, and
    \item Nodes in $r_{x + 1, y}$ map to nodes in $T_2[j, \Right(v) - 1)$, contributing $s_{j, \Right(v) - 1}(r_{x + 1, y})$.
\end{itemize}
The total contribution is $\eta(u_x, v) + s_{\Left(v) + 1, i}(l_{x + 1, y}) + {\hat s}_{i, j}(\sub(u_y)) + s_{j, \Right(v) - 1}(r_{x + 1, y})$, and the contribution applies to ${\hat s}_{i ^ {\prime}, j ^ {\prime}}(\sub(u_x))$ for all $i ^ {\prime} \le \Left(v) < \Right(v) \le j ^ {\prime}$.

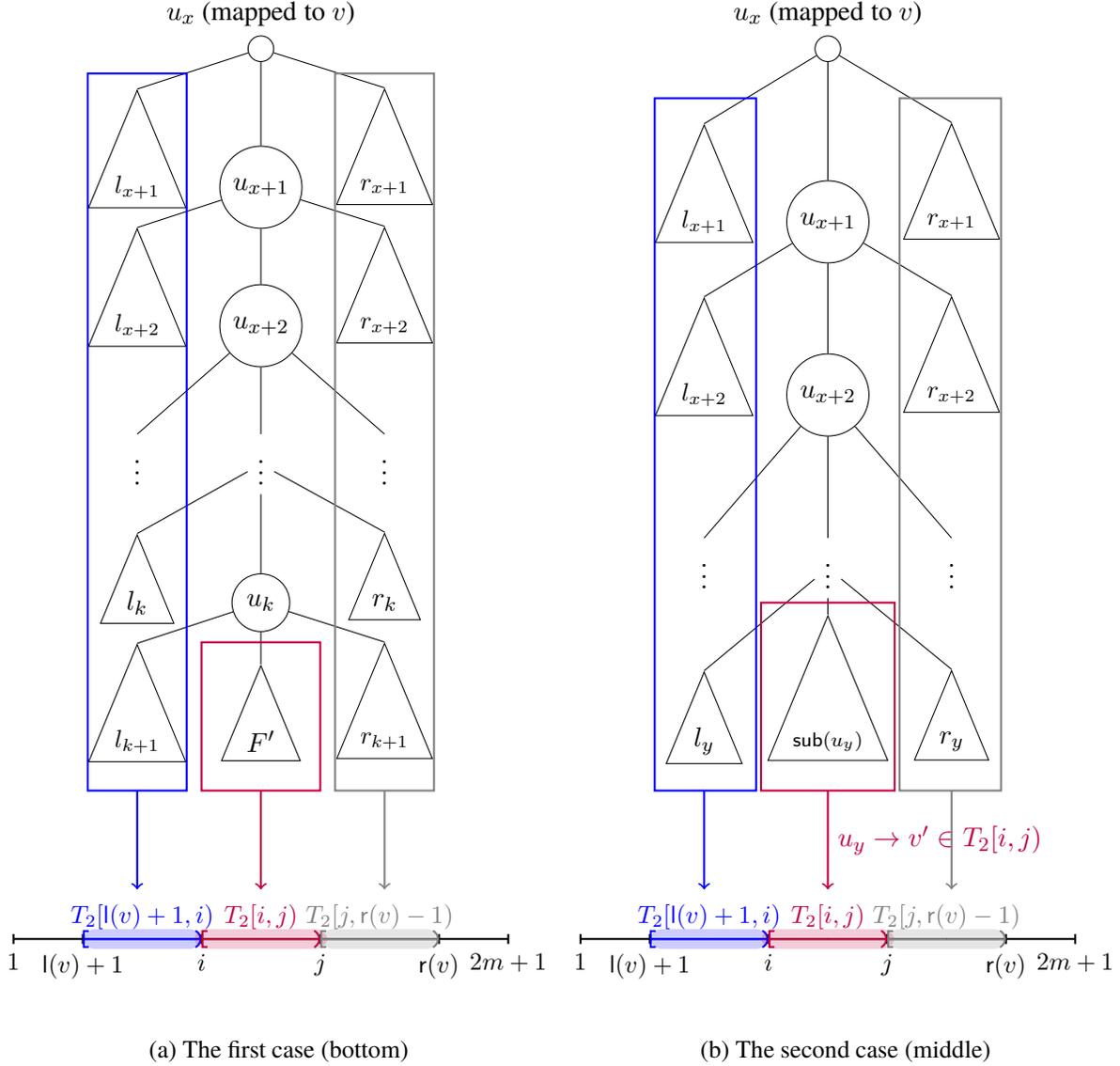
\begin{figure} [H]  
\centering 
    \begin{subfigure}[b]{0.48\linewidth}
        \begin{tikzpicture}[
scale=0.7,
every node/.style={solid},
edge from parent/.style={draw,solid}, 
level/.style={sibling distance=25mm},
level distance=28mm
]
\node [circle,draw,label=above:{$u_x$ (mapped to $v$)}] {}
    [child anchor=north]
    child {node [triangle,draw] {\small $l_{x + 1}$} edge from parent[solid]
    }
    child {node [circle,draw] {$u_{x + 1}$} edge from parent[solid]
        child {node [triangle,draw] {\small $l_{x + 2}$} 
        }
        child {node [circle,draw] {$u_{x + 2}$}  edge from parent[solid]
            child {node {$\vdots$}  edge from parent[solid]
            }
            child {node {$\vdots$}  edge from parent[solid]
                child {node [triangle,draw] {$l_k$} 
                }
                child {node [circle,draw] {$u_k$}  edge from parent[solid]
                    child {node [triangle,draw] {\small $l_{k + 1}$} 
                    }
                    child {node [triangle,draw] {$F ^ {\prime}$}
                    }
                    child {node [triangle,draw] {\small $r_{k + 1}$}
                    };
                }
                child {node [triangle,draw] {$r_k$}
                };
            }
            child {node {$\vdots$}  edge from parent[solid]
            };
        }
        child {node [triangle,draw] {\small $r_{x + 2}$}
        };
    }
    child {node [triangle,draw] {\small $r_{x + 1}$} edge from parent[solid]
    };
\draw[thick] (-5,-18) -- (5,-18);
\foreach \x/\xtext in {-5/$1$,-3.6/$\Left(v) + 1$,-1.2/$i$,1.2/$j$,3.6/$\Right(v)$, 5/$2m + 1$}
    \draw[thick] (\x,-17.9) -- (\x,-18.1) node[below] {\small \xtext};
    
\draw[[-, thick, blue] (-3.6,-18) -- (-2.4,-18)  node[above] {\small $T_2[\Left(v) + 1, i)$};
\draw[-), thick, blue] (-2.4,-18) -- (-1.2,-18);
\fill[opacity = 0.2, blue,rounded corners=1ex] (-3.6,-17.8) -- (-1.2, -17.8) -- (-1.2, -18.2) -- (-3.6, -18.2) -- cycle;

\draw[thick, blue] (-3.5, -0.5) -- (-3.5, -15) -- (-1.5, -15) -- (-1.5, -0.5) -- cycle; 
\draw[->, thick, blue] (-2.5, -15) -- (-2.5, -17);

\draw[[-, thick, purple] (-1.2,-18) -- (0,-18)  node[above] {\small $T_2[i, j)$};
\draw[-), thick, purple] (0,-18) -- (1.2,-18);
\fill[opacity = 0.2, purple,rounded corners=1ex] (-1.35,-17.8) -- (1.35, -17.8) -- (1.35, -18.2) -- (-1.35, -18.2) -- cycle;

\draw[thick, purple] (-1.2, -12) -- (-1.2, -15) -- (1.2, -15) -- (1.2, -12) -- cycle;
\draw[->, thick, purple] (0, -15) -- (0, -17);

\draw[[-, thick, gray] (1.2,-18) -- (2.4,-18)  node[above] {\small $T_2[j, \Right(v) - 1)$};
\draw[-), thick, gray] (2.4,-18) -- (3.6,-18);
\fill[opacity = 0.2, gray,rounded corners=1ex] (1.2,-17.8) -- (3.6, -17.8) -- (3.6, -18.2) -- (1.2, -18.2) -- cycle;

\draw[thick, gray] (1.5, -0.5) -- (1.5, -15) -- (3.5, -15) -- (3.5, -0.5) -- cycle;
\draw[->, thick, gray] (2.5, -15) -- (2.5, -17);

\end{tikzpicture}
        \caption{The first case (bottom)} \label{fig:bottom}  
    \end{subfigure}
    \begin{subfigure}[b]{0.48\linewidth}
        \begin{tikzpicture}[
scale=0.7,
every node/.style={solid},
edge from parent/.style={draw,solid}, 
level/.style={sibling distance=25mm},
level distance=35mm
]
\node [circle,draw,label=above:{$u_x$ (mapped to $v$)}] {}
    [child anchor=north]
    child {node [triangle,draw] {\small $l_{x + 1}$} edge from parent[solid]
    }
    child {node [circle,draw] {$u_{x + 1}$} edge from parent[solid]
        child {node [triangle,draw] {\small $l_{x + 2}$} 
        }
        child {node [circle,draw] {$u_{x + 2}$}  edge from parent[solid]
            child {node {$\vdots$}  edge from parent[solid]
            }
            child {node {$\vdots$}  edge from parent[solid]
                child {node [triangle,draw] {$l_y$} 
                }
                child {node [triangle,draw] {\scriptsize $\sub(u_y)$}  edge from parent[solid]
                }
                child {node [triangle,draw] {$r_y$}
                };
            }
            child {node {$\vdots$}  edge from parent[solid]
            };
        }
        child {node [triangle,draw] {\small $r_{x + 2}$}
        };
    }
    child {node [triangle,draw] {\small $r_{x + 1}$} edge from parent[solid]
    };
\draw[thick] (-5,-18) -- (5,-18);
\foreach \x/\xtext in {-5/$1$,-3.6/$\Left(v) + 1$,-1.2/$i$,1.2/$j$,3.6/$\Right(v)$, 5/$2m + 1$}
    \draw[thick] (\x,-17.9) -- (\x,-18.1) node[below] {\small \xtext};
    
\draw[[-, thick, blue] (-3.6,-18) -- (-2.4,-18)  node[above] {\small $T_2[\Left(v) + 1, i)$};
\draw[-), thick, blue] (-2.4,-18) -- (-1.2,-18);
\fill[opacity = 0.2, blue,rounded corners=1ex] (-3.6,-17.8) -- (-1.2, -17.8) -- (-1.2, -18.2) -- (-3.6, -18.2) -- cycle;

\draw[thick, blue] (-3.5, -1) -- (-3.5, -15) -- (-1.45, -15) -- (-1.45, -1) -- cycle; 
\draw[->, thick, blue] (-2.5, -15) -- (-2.5, -17);

\draw[[-, thick, purple] (-1.2,-18) -- (0,-18)  node[above] {\small $T_2[i, j)$};
\draw[-), thick, purple] (0,-18) -- (1.2,-18);
\fill[opacity = 0.2, purple,rounded corners=1ex] (-1.35,-17.8) -- (1.35, -17.8) -- (1.35, -18.2) -- (-1.35, -18.2) -- cycle;

\draw[thick, purple] (-1.35, -11.2) -- (-1.35, -15) -- (1.35, -15) -- (1.35, -11.2) -- cycle;
\draw[->, thick, purple] (0, -15) -- (0, -17) node[midway, right, align=center] { $u_y \rightarrow v ^ {\prime} \in T_2[i, j)$};

\draw[[-, thick, gray] (1.2,-18) -- (2.4,-18)  node[above] {\small $T_2[j, \Right(v) - 1)$};
\draw[-), thick, gray] (2.4,-18) -- (3.6,-18);
\fill[opacity = 0.2, gray,rounded corners=1ex] (1.2,-17.8) -- (3.6, -17.8) -- (3.6, -18.2) -- (1.2, -18.2) -- cycle;

\draw[thick, gray] (1.45, -1) -- (1.45, -15) -- (3.5, -15) -- (3.5, -1) -- cycle;
\draw[->, thick, gray] (2.5, -15) -- (2.5, -17);

\end{tikzpicture}
        \caption{The second case (middle)} \label{fig:mid}  
    \end{subfigure}
\caption{Computing $\hat S(\sub(u_x))$}
\end{figure}

Using an idea similar to the one in Algorithm \ref{algorithm:mulbasic}, the number of triples in the first case can be reduced to $O(\abs{T_2}(\abs{F} - \abs{F ^ {\prime}}) ^ 2)$, while the number of quadruples in the second case can be reduced to $O(\abs{T_2}(\abs{F} - \abs{F ^ {\prime}}) ^ 3)$, which is just an extra $O(\abs{F} - \abs{F ^ {\prime}}))$ factor more than the bound in Lemma \ref{lemma:transition2}. The details are shown in Algorithm \ref{algorithm:transition2}. In fact, this extra $O(\abs{F} - \abs{F ^ {\prime}}))$ overhead in quadruple count will only change the $\abs{T_2}(\abs{F} - \abs{F ^ {\prime}}) ^ 3$ term in Theorem \ref{theorem:transition} to $\abs{T_2}(\abs{F} - \abs{F ^ {\prime}}) ^ 4$, and one can still achieve a truly sub-cubic running time overall by setting the block size $\Delta = O(\abs{T_1} ^ c)$ for any $0 < c < \frac{1}{3}$.

\begin{algorithm} [H]
    \caption{Computation of $\hat S(\sub(u_x))$ from $\hat S(\sub(u_y))$ for $y > x$} \label{algorithm:transition2}
    \begin{algorithmic}[1]
        \State $\hat S(\sub(u_x)) \gets [-\infty]_{2\abs{T_2} + 1, 2\abs{T_2} + 1}$
        \For {$v \in T_2$} \Comment{The first case (bottom)}
            \For {$z \in [0, 2(\abs{F} - \abs{F ^ {\prime}})]$}
                \For {$w \in [0, 2(\abs{F} - \abs{F ^ {\prime}})]$}
                    \State $i \gets \mincol(S(l_{x + 1, k + 1}), \Left(v) + 1, z)$
                    \State $j \gets \maxrow(S(r_{x + 1, k + 1}), \Right(v) - 1, w)$
                    \State $t \gets \eta(u_x, v) + s_{\Left(v) + 1, i}(l_{x + 1, k + 1}) + s_{i, j}(F ^ {\prime}) + s_{j, \Right(v) - 1}(r_{x + 1, k + 1})$
                    \State $\hat S(\sub(u_x)) \gets \rangemax(\hat S(\sub(u_x)), \Left(v), \Right(v), t)$
                \EndFor
            \EndFor
        \EndFor
        \For {$v \in T_2$} \Comment{The second case (middle)}
            \For {$z \in [0, 2(\abs{F} - \abs{F ^ {\prime}})]$}
                \For {$w \in [0, 2(\abs{F} - \abs{F ^ {\prime}})]$}
                    \For {$y \in [x + 1, k]$}
                        \State $i \gets \mincol(S(l_{x + 1, y}), \Left(v) + 1, z)$
                        \State $j \gets \maxrow(S(r_{x + 1, y}), \Right(v) - 1, w)$
                        \State $t \gets \eta(u_x, v) + s_{\Left(v) + 1, i}(l_{x + 1, y}) + {\hat s}_{i, j}(\sub(u_y)) + s_{j, \Right(v) - 1}(r_{x + 1, y})$
                        \State $\hat S(\sub(u_x)) \gets \rangemax(\hat S(\sub(u_x)), \Left(v), \Right(v), t)$
                    \EndFor
                \EndFor
            \EndFor
        \EndFor
    \end{algorithmic}
\end{algorithm}

\paragraph{Speed-up for the second case.} To speed up the running time for the second case to $\tilde O(\abs{T_2}(\abs{F} - \abs{F ^ {\prime}}) ^ 2)$, we note that on $T_2$, $v$ must be an ancestor of the node $v ^ {\prime}$ that $u_y$ maps to, and the contribution to ${\hat S(\sub(u_x))}$ given $u_x$ mapping to $v$ and $u_y$ mapping to $v ^ {\prime}$ is 
\begin{align*}
    \eta(u_x, v) + s_{\Left(v) + 1, \Left(v ^ {\prime})}(l_{x + 1, y}) + \hat s_{\Left(v ^ {\prime}), \Right(v ^ {\prime})}(\sub(u_y)) + s_{\Right(v ^ {\prime}), \Right(v) - 1}(r_{x + 1, y}).
\end{align*}
We let 
\begin{align*}
    \rho(u_y, v ^ {\prime}, v) &= s_{\Left(v) + 1, \Left(v ^ {\prime})}(l_{x + 1, y}) + \hat s_{\Left(v ^ {\prime}), \Right(v ^ {\prime})}(\sub(u_y)) + s_{\Right(v ^ {\prime}), \Right(v) - 1}(r_{x + 1, y})
\end{align*}
and 
\begin{align*}
    \phi(u_y, v) &= \max\{\rho(u_y, v ^ {\prime}, v) \mid v ^ {\prime} \in \sub(v)\}.
\end{align*}
If $\phi(u_y, v)$ is known for all $u_y$ and $v$, computing the contribution to ${\hat S}(\sub(u_x))$ takes only $O(\abs{T_2}(\abs{F} - \abs{F ^ {\prime}}))$ time, which is not a bottleneck. 

Consider enumerating $u_y$ and $v ^ {\prime}$ and updating $\phi(u_y, v)$ for each $v$ on the path from $v ^ {\prime}$ to $\Root(T_2)$ using $\rho(u_y, v ^ {\prime}, v)$. For fixed $u_y$ and $v ^ {\prime}$, $\rho(u_y, v ^ {\prime}, v) = {\hat s}_{\Left(v ^ {\prime}), \Right(v ^ {\prime})}(\sub(u_y)) + f(v)$ for some function $f(v)$ with value in $[0, 2(\abs{F} - \abs{F ^ {\prime}})]$. Moreover, if we move $v$ from $v ^ {\prime}$ to $\Root(T_2)$, $f(v)$ is non-decreasing, and therefore only changes $O(\abs{F} - \abs{F ^ {\prime}})$ times. It is easy to find all the endpoints for the changes in $\tilde O(\abs{F} - \abs{F ^ {\prime}})$ time. The updates to $\phi(u_y, v)$ now consist of $O(\abs{F} - \abs{F ^ {\prime}})$ path modifications, and can be done with a data structure such as a \emph{link/cut tree} \cite{SLEATOR1983362} in $\tilde O(\abs{F} - \abs{F ^ {\prime}})$ time. Since we need to enumerate $u_y$ and $v ^ {\prime}$, the total running time is $\tilde O(\abs{T_2}(\abs{F} - \abs{F ^ {\prime}}) ^ 2)$, which meets the bound in Lemma \ref{lemma:transition2}.

\subsection{Faster max-plus multiplications for special bounded-difference matrices} \label{sec:mul}
In this section we show that Theorem \ref{theorem:mul} is true, thereby verifying that the type I transitions can be done within the desired time bound. 

We will first show an algorithm that proves a weaker version of Theorem \ref{theorem:mul} where no $-\infty$ below the main diagonal of $A$ is present. For matrices $M_1, M_2$ of the same dimensions, let $\max(M_1, M_2)$ be the entry-wise maximum of $M_1$ and $M_2$.
\begin{lemma} \label{lemma:mulmonotoneweaker}
    Let $A$ be an $l \times l$ row-monotone, column-monotone matrix whose entries are integers in $[0, m]$. Let $B$ be an $l \times n$ row-monotone, column-monotone matrix. Let $C ^ {\prime}$ be any row-monotone, column-monotone $l \times n$ matrix. Then $C = \max(C ^ {\prime}, A \star B)$ can be computed in $\tilde O(m ^ 2n)$ time (under the range operation/query model from Section \ref{sec:model}).
\end{lemma}

To prove Lemma \ref{lemma:mulmonotoneweaker}, we can use Algorithm \ref{algorithm:mulweaker}.
\begin{algorithm} [H]
    \caption{Computation of $C = \max(C ^ {\prime}, A \star B)$ in Lemma \ref{lemma:mulmonotoneweaker}} \label{algorithm:mulweaker}
    \begin{algorithmic}[1]
        \Procedure{MUL2}{$A, B, C ^ {\prime}$}
            \State $C \gets C ^ {\prime}$
            \For {$j \in [1, n]$}
                \For {$x \in [B_{1, j} - m, B_{1, j}]$} \label{line:xrange}
                    \For {$y \in [0, m]$}
                        \State $k \gets \maxrow(B, j, x)$
                        \State $i \gets \maxrow(A, k, y)$
                        \State $C \gets \rangemax(C, i, j, A_{i, k} + B_{k, j})$
                    \EndFor
                \EndFor
            \EndFor
        \EndProcedure
    \end{algorithmic}
\end{algorithm}
The only major difference between Algorithm \ref{algorithm:mulweaker} and Algorithm \ref{algorithm:mulbasic} is on Line \ref{line:xrange}: we enumerate $x$ only over the $m + 1$ largest possible values on column $j$. Note that a triple $(i, k, j)$ cannot be useful when $B_{k, j} < B_{1, j} - m$ since 
\begin{align*}
    A_{i, k} + B_{k, j} \le m + B_{k, j} < B_{1, j} \le A_{i, 1} + B_{1, j}.
\end{align*}
The correctness of Algorithm \ref{algorithm:mulweaker} is now straightforward.

We now prove Theorem \ref{theorem:mul} by designing a recursive process that computes $C ^ {\prime} = A ^ {\prime} \star B ^ {\prime}$ for $l \times l$ sub-matrix $A ^ {\prime}$ of $A$ and $l \times n$ sub-matrix $B ^ {\prime}$ of $B$. Suppose $l$ is a power of 2. We partition $A ^ {\prime}$ into four $l / 2 \times l / 2$ matrices and $B ^ {\prime}$ into two $l / 2 \times n$ matrices
\begin{align*}
    A ^ {\prime} &= 
        \begin{pmatrix}
             D & E\\
             -\infty & F
        \end{pmatrix}, \\
    B ^ {\prime} &= 
        \begin{pmatrix}
             G \\
             H
        \end{pmatrix}.
\end{align*}
Then 
\begin{align*}
    A ^ {\prime} \star B ^ {\prime} = 
        \begin{pmatrix}
             \max(D \star G, E \star H) \\
             F \star H
        \end{pmatrix}.
\end{align*}
Since the sub-matrix $E$ does not contain $-\infty$, $E \star H$ can be computed using Algorithm \ref{algorithm:mulweaker}. We now reduce the problem into two smaller problems of size $l / 2$. Algorithm \ref{algorithm:mul} describes how to multiply $A ^ {\prime}$ and $B ^ {\prime}$ recursively based on this idea.
\begin{algorithm} [H]
    \caption{Computation of $C ^ {\prime} = A ^ {\prime} \star B ^ {\prime}$ for $l \times l$ sub-matrix $A ^ {\prime}$ of $A$ and $l \times n$ sub-matrix $B ^ {\prime}$ of $B$ in Theorem \ref{theorem:mul}}  \label{algorithm:mul}
    \begin{algorithmic}[1]
        \Procedure{MUL3}{$A ^ {\prime}, B ^ {\prime}$}
            \If {$l > \Delta$}
                \State $D, E, F \gets {A ^ {\prime}}_{[1, l / 2], [1, l / 2]}, {A ^ {\prime}}_{[1, l / 2], [l / 2 + 1, l]}, {A ^ {\prime}}_{[l / 2 + 1, l], [l / 2 + 1, l]}$
                \State $G, H \gets {B ^ {\prime}}_{[1, l / 2], [1, n]}, {B ^ {\prime}}_{[l / 2 + 1, l], [1, n]}$
                \State ${C ^ {\prime}}_{\textrm{up}} \gets $MUL$_2$($E$, $H$, MUL$_3$($D$, $G$))
                \State ${C ^ {\prime}}_{\textrm{down}} \gets $MUL$_3$($F$, $H$))
                \State ${C ^ {\prime}} \gets \begin{pmatrix} C_{\textrm{up}} \\ C_{\textrm{down}} \end{pmatrix}$
            \Else
                \State Calculate $C ^ {\prime} = {A ^ {\prime}} \star {B ^ {\prime}}$ using Theorem \ref{theorem:bringmann} \label{line:bringmann}
            \EndIf
        \EndProcedure
    \end{algorithmic}
\end{algorithm}

It may seem that we cannot use Theorem \ref{theorem:bringmann} on line \ref{line:bringmann} since $A ^ {\prime}$ and $B ^ {\prime}$ may contain $-\infty$, which is not allowed in Theorem \ref{theorem:bringmann}. To resolve this issue, before calling Algorithm \ref{algorithm:mul} on $A$ and $B$, we find $W$ which is the largest difference between adjacent non-$(-\infty)$ entries in $A$ and $B$ (for similarity matrices we can simply let $W = 2$) and set every $(i, j)$-entry below the main diagonals of $A$ and $B$ to $W(j - i)$. This does not affect the answer on or above the main diagonal. Since the entries on the main diagonals of $A$ and $B$ are zero, in all the recursive calls $A ^ {\prime}$ and $B ^ {\prime}$ are $W$-bounded-difference matrices.

We now analyze the running time of computing $C = A \star B$ using Algorithm \ref{algorithm:mul}. If we ignore Line \ref{line:bringmann}, then each recursive call contributes $\tilde O(m ^ 2n)$ to the running time. Since we stop the recursion at $l \le \Delta$, the total number of recursive calls is $O(n / \Delta)$. Therefore the total running time for this part is $\tilde O(n ^ 2m ^ 2 / \Delta)$. Line \ref{line:bringmann} is executed $O(n / \Delta)$ times. The time needed for multiplying two bounded-difference matrices of sizes $\Delta \times \Delta$ and $\Delta \times n$ does not exceed $n / \Delta$ times the time needed for multiplying two $\Delta \times \Delta$ bounded-difference matrices, and by applying Theorem \ref{theorem:bringmann} we get $\tilde O(n\Delta ^ {1.8244})$ randomized and $\tilde O(n\Delta ^ {1.8603})$ deterministic time per multiplication. All $O(n / \Delta)$ multiplications together take $\tilde O(n ^ 2\Delta ^ {0.8244})$ randomized and $\tilde O(n ^ 2\Delta ^ {0.8603})$ deterministic time. By setting $\Delta = \tilde O(m ^ {1.0963})$ in the randomized case and $\Delta = \tilde O(m ^ {1.0751})$ in the deterministic case we get the bounds in Theorem \ref{theorem:mul}. Our analysis breaks when $\Delta > n$, but in this case, Algorithm \ref{algorithm:mul} is equivalent to directly applying the algorithm from Theorem \ref{theorem:bringmann}, and one can verify that it still achieves the bounds in Theorem \ref{theorem:mul}. 

\section{Potential speed-ups and future work} \label{sec:futurework}

\subsection{Potential speed-ups}
We note that in Section \ref{sec:mul}. The bounded-difference matrix multiplciation in \cite{Bringmann16} is applied to rectangular matrices by repeated calling it on square matrices as a black box. By closely examining this multiplication algorithm, one can adapt it to the rectangular case in a more efficient way. The algorithm consists of three phases, and the second phase relies on ordinary matrix multiplication between rectangular matrices. Optimization is possible since rectangular matrix multiplication can be done more efficiently than repeated square matrix multiplications \cite{legall2012}. Moreover, the analysis of the running time of the third phase of their algorithm is reliant on a result on balanced bipartite graphs (Lemma 4 of \cite{Bringmann16}). A similar result would also hold for unbalanced bipartite graphs, and can be used to better bound the running time for the rectangular case. These speed-ups can lead to slight improvement in the exponents in our running time.

The algorithm in \cite{Bringmann16} did not use the row-monotone and column-monotone properties of the matrices involved. We note that Vassilevska Williams and Xu extended the algorithm to less structured matrices \cite{xu2020}, and one of their results showed a truly sub-cubic algorithm for \emph{Monotone Min-plus Product}, which was later improved in \cite{gu2021faster}. These results showed that the min-plus product of $n \times n$ matrices $A$ and $B$ where entries of $B$ are non-negative integers bounded by $O(n)$ and each row of $B$ is non-decreasing can be deterministically computed in $O(n ^ {2.8653})$ time (no bounded-difference property required!). Therefore, there is potential that even better time bound for tree edit distance can be achieved by better exploiting the properties of similarity matrices when doing the max-plus products.

\subsection{Future work}
There are still many remaining open problems related to tree edit distance. The most valuable one is perhaps whether weakly sub-cubic\footnote{By weakly sub-cubic we mean sub-cubic but not truly sub-cubic (i.e. by sub-polynomial factors).} algorithms can be found for the weighted tree edit distance problem. 

The conditional hardness results in \cite{apsphard2020} were based on APSP, but weakly sub-cubic algorithms for APSP have long been known, and the current best algorithm achieves the same $n ^ 3 / 2 ^ {\Omega(\sqrt{\log {n}})}$ running time as the algorithm for min-plus product by Williams \cite{rrwapsp2014}, so a natural direction for sub-polynomial improvement to the cubic running time for weighted tree edit distance is to find a reduction from weighted tree edit distance to APSP. It is tempting to think that one can adapt our algorithm to the weighted setting to obtain a weakly sub-cubic running time simply by replacing the specialized algorithm in \cite{Bringmann16} with the more general algorithm by Williams. Unfortunately, this is not the case. Firstly, our algorithm involves more than $O(n ^ {0.5})$ max-plus products. It works in the unweighted setting since most of the time the entries in one of the matrices involved are small and the product can be computed very efficiently. This is no longer the case in the weighted setting. Secondly, some other components of our algorithm have a running time dependent on weights. When the weights become large enough, the running time of these components becomes unacceptable. 

Another direction of research is to design better algorithms for bounded tree edit distance. When the tree edit distance is bounded by $k$, the current best algorithm runs in $O(nk ^ 2\min(k, \log n))$ time by combining the results of Akmal and Jin \cite{akmal2021faster} and Touzet  \cite{Touzet05, Touzet07}, which when $k = O(n)$ gives us a running time of $O(n ^ 3 \log n)$. Previous methods achieve their running time by pruning useless states in their dynamic programming schemes. In Touzet's method, this is done using the pre-order traversal sequence, which also seems to work on the bi-order traversal sequence our algorithm is based on. Thus it is possible that some adaptation of our algorithm to the bounded case can give us an $\tilde O(nk ^ {2 - \eps})$ running time for some $\eps > 0$. Moreover, since the easier problem of bounded \emph{string} edit distance can be solved in $\tilde O(n + k ^ 2)$ time \cite{Meyers86,LandauV88}, it is possible that bounded tree edit distance admits an $\tilde O(n + k ^ c)$ algorithm for some constant $c$. Since (unweighted) tree edit distance in general can be solved in truly sub-cubic time, an ideal value of $c$ would be less than 3.

In the unrooted setting, the most recent algorithm was by Dudek and Gawrychowski \cite{unrooted2018}, which improved upon the previous best $O(n^3 \log n)$ time bound due to Klein \cite{Klein98}, and achieved the same $O(n ^ 3)$ running time as the algorithm for the rooted case in \cite{demaine2007}. It will also be interesting to see if a truly sub-cubic running time is possible for this unrooted setting. Moreover, when the distance is bounded by $k$, unlike the case in the rooted setting, there has not been an algorithm with run-time quasi-linear in $n$ (i.e. $\tilde O(nk ^ c)$ for some constant $c$). The algorithm by Akmal and Ce \cite{akmal2021faster}, for example, only runs in quadratic time when $k = O(1)$. This gives us another area where potential future work can be done.

\section*{Acknowledgment}
I want to thank Ce Jin for introducing tree edit distance to me.

I want to thank Professor Virginia Vassilevska Williams and Yinzhan Xu from Massachusetts Institute of Technology for giving me help on the matrix multiplication part of the algorithm. 

Last but not least, I want to thank Professor Virginia Vassilevska Williams, Ce Jin, Yinzhan Xu and Lijie Chen from Massachusetts Institute of Technology and Zhaoyi Hao from University of New South Wales for giving me advice on my write-up.

\bibliography{theory}

\newcommand{\etalchar}[1]{$^{#1}$}
\begin{thebibliography}{BGSVW16}

\bibitem[AALM90]{doi:10.1137/0219066}
Alberto Apostolico, Mikhail~J. Atallah, Lawrence~L. Larmore, and Scott
  McFaddin.
\newblock Efficient parallel algorithms for string editing and related
  problems.
\newblock {\em SIAM Journal on Computing}, 19(5):968--988, 1990.
\newblock \href {http://arxiv.org/abs/https://doi.org/10.1137/0219066}
  {\path{arXiv:https://doi.org/10.1137/0219066}}, \href
  {https://doi.org/10.1137/0219066} {\path{doi:10.1137/0219066}}.

\bibitem[ABV15]{AbboudBW15}
Amir Abboud, Arturs Backurs, and Virginia {Vassilevska Williams}.
\newblock Tight hardness results for {LCS} and other sequence similarity
  measures.
\newblock In {\em Proceedings of the 56th {IEEE} Symposium on Foundations of
  Computer Science {(FOCS)}}, pages 59--78, 2015.
\newblock \href {https://doi.org/10.1109/FOCS.2015.14}
  {\path{doi:10.1109/FOCS.2015.14}}.

\bibitem[AJ21]{akmal2021faster}
Shyan Akmal and Ce~Jin.
\newblock Faster algorithms for bounded tree edit distance, 2021.
\newblock To appear in ICALP 2021.
\newblock \href {http://arxiv.org/abs/2105.02428} {\path{arXiv:2105.02428}}.

\bibitem[AKM{\etalchar{+}}86]{10.1145/10515.10546}
A~Aggarwal, M~Klawe, S~Moran, P~Shor, and R~Wilber.
\newblock Geometric applications of a matrix searching algorithm.
\newblock In {\em Proceedings of the Second Annual Symposium on Computational
  Geometry}, SCG '86, page 285–292, New York, NY, USA, 1986. Association for
  Computing Machinery.
\newblock \href {https://doi.org/10.1145/10515.10546}
  {\path{doi:10.1145/10515.10546}}.

\bibitem[AKO10]{AndoniKO10}
Alexandr Andoni, Robert Krauthgamer, and Krzysztof Onak.
\newblock Polylogarithmic approximation for edit distance and the asymmetric
  query complexity.
\newblock In {\em Proceedings of the 51st Annual {IEEE} Symposium on
  Foundations of Computer Science (FOCS)}, pages 377--386, 2010.
\newblock \href {https://doi.org/10.1109/FOCS.2010.43}
  {\path{doi:10.1109/FOCS.2010.43}}.

\bibitem[AN20]{AndoniN20}
Alexandr Andoni and Negev~Shekel Nosatzki.
\newblock Edit distance in near-linear time: it's a constant factor.
\newblock In {\em Proceedings of the 61st {IEEE} Annual Symposium on
  Foundations of Computer Science (FOCS)}, 2020.

\bibitem[AO12]{AndoniO12}
Alexandr Andoni and Krzysztof Onak.
\newblock Approximating edit distance in near-linear time.
\newblock {\em {SIAM} J. Comput.}, 41(6):1635--1648, 2012.
\newblock \href {https://doi.org/10.1137/090767182}
  {\path{doi:10.1137/090767182}}.

\bibitem[BEG{\etalchar{+}}18]{BoroujeniEGHS18}
Mahdi Boroujeni, Soheil Ehsani, Mohammad Ghodsi, Mohammad~Taghi Hajiaghayi, and
  Saeed Seddighin.
\newblock Approximating edit distance in truly subquadratic time: Quantum and
  {MapReduce}.
\newblock In {\em Proceedings of the 29th Annual {ACM-SIAM} Symposium on
  Discrete Algorithms (SODA)}, pages 1170--1189, 2018.
\newblock \href {https://doi.org/10.1137/1.9781611975031.76}
  {\path{doi:10.1137/1.9781611975031.76}}.

\bibitem[BF08]{tcs/BilleF08}
Philip Bille and Martin Farach{-}Colton.
\newblock Fast and compact regular expression matching.
\newblock {\em Theor. Comput. Sci.}, 409(3):486--496, 2008.
\newblock \href {https://doi.org/10.1016/j.tcs.2008.08.042}
  {\path{doi:10.1016/j.tcs.2008.08.042}}.

\bibitem[BGHS19]{apxfocs2019}
Mahdi Boroujeni, Mohammad Ghodsi, MohammadTaghi Hajiaghayi, and Saeed
  Seddighin.
\newblock 1+{{\(\epsilon\)}} approximation of tree edit distance in quadratic
  time.
\newblock In {\em Proceedings of the 51st Annual {ACM} {SIGACT} Symposium on
  Theory of Computing ({STOC})}, pages 709--720, 2019.
\newblock \href {https://doi.org/10.1145/3313276.3316388}
  {\path{doi:10.1145/3313276.3316388}}.

\bibitem[BGK03]{KochBG03}
Peter Buneman, Martin Grohe, and Christoph Koch.
\newblock Path queries on compressed {XML}.
\newblock In {\em Proceedings of the 29th International Conference on Very
  Large Data Bases ({VLDB})}, pages 141--152, 2003.
\newblock \href {https://doi.org/10.1016/B978-012722442-8/50021-5}
  {\path{doi:10.1016/B978-012722442-8/50021-5}}.

\bibitem[BGMW20]{apsphard2020}
Karl Bringmann, Pawe\l{} Gawrychowski, Shay Mozes, and Oren Weimann.
\newblock Tree edit distance cannot be computed in strongly subcubic time
  (unless {APSP} can).
\newblock {\em {ACM} Trans. Algorithms}, 16(4):48:1--48:22, 2020.
\newblock \href {https://doi.org/10.1145/3381878} {\path{doi:10.1145/3381878}}.

\bibitem[BGSVW16]{Bringmann16}
Karl Bringmann, Fabrizio Grandoni, Barna Saha, and Virginia
  Vassilevska~Williams.
\newblock Truly sub-cubic algorithms for language edit distance and rna-folding
  via fast bounded-difference min-plus product.
\newblock In {\em 2016 IEEE 57th Annual Symposium on Foundations of Computer
  Science (FOCS)}, pages 375--384, 2016.
\newblock \href {https://doi.org/10.1109/FOCS.2016.48}
  {\path{doi:10.1109/FOCS.2016.48}}.

\bibitem[BI18]{BackursI18}
Arturs Backurs and Piotr Indyk.
\newblock Edit distance cannot be computed in strongly subquadratic time
  (unless {SETH} is false).
\newblock {\em {SIAM} J. Comput.}, 47(3):1087--1097, 2018.
\newblock \href {https://doi.org/10.1137/15M1053128}
  {\path{doi:10.1137/15M1053128}}.

\bibitem[Bil05]{Bille05survey}
Philip Bille.
\newblock A survey on tree edit distance and related problems.
\newblock {\em Theor. Comput. Sci.}, 337(1-3):217--239, 2005.
\newblock \href {https://doi.org/10.1016/j.tcs.2004.12.030}
  {\path{doi:10.1016/j.tcs.2004.12.030}}.

\bibitem[BK99]{BellandoK99}
John Bellando and Ravi Kothari.
\newblock Region-based modeling and tree edit distance as a basis for gesture
  recognition.
\newblock In {\em Proceedings of the 10th International Conference on Image
  Analysis and Processing {(ICIAP)}}, pages 698--703. {IEEE} Computer Society,
  1999.
\newblock \href {https://doi.org/10.1109/ICIAP.1999.797676}
  {\path{doi:10.1109/ICIAP.1999.797676}}.

\bibitem[BR20]{BrakensiekR20}
Joshua Brakensiek and Aviad Rubinstein.
\newblock Constant-factor approximation of near-linear edit distance in
  near-linear time.
\newblock In {\em Proccedings of the 52nd Annual {ACM} {SIGACT} Symposium on
  Theory of Computing (STOC)}, pages 685--698, 2020.
\newblock \href {https://doi.org/10.1145/3357713.3384282}
  {\path{doi:10.1145/3357713.3384282}}.

\bibitem[CDG{\etalchar{+}}18]{ChakrabortyDGKS18}
Diptarka Chakraborty, Debarati Das, Elazar Goldenberg, Michal Kouck{\'{y}}, and
  Michael~E. Saks.
\newblock Approximating edit distance within constant factor in truly
  sub-quadratic time.
\newblock In {\em Proceedings of the 59th {IEEE} Annual Symposium on
  Foundations of Computer Science ({FOCS})}, pages 979--990. {IEEE} Computer
  Society, 2018.
\newblock \href {https://doi.org/10.1109/FOCS.2018.00096}
  {\path{doi:10.1109/FOCS.2018.00096}}.

\bibitem[Cha99]{Chawathe99}
Sudarshan~S. Chawathe.
\newblock Comparing hierarchical data in external memory.
\newblock In {\em Proceedings of the 25th International Conference on Very
  Large Data Bases ({VLDB})}, pages 90--101, 1999.
\newblock URL: \url{http://www.vldb.org/conf/1999/P8.pdf}.

\bibitem[Che01]{chen01}
Weimin Chen.
\newblock New algorithm for ordered tree-to-tree correction problem.
\newblock {\em Journal of Algorithms}, 40(2):135--158, 2001.
\newblock \href {https://doi.org/https://doi.org/10.1006/jagm.2001.1170}
  {\path{doi:https://doi.org/10.1006/jagm.2001.1170}}.

\bibitem[DG18]{unrooted2018}
Bart\l{}omiej Dudek and Pawe\l{} Gawrychowski.
\newblock Edit distance between unrooted trees in cubic time.
\newblock In {\em Proceedings of the 45th International Colloquium on Automata,
  Languages, and Programming ({ICALP})}, pages 45:1--45:14, 2018.
\newblock \href {https://doi.org/10.4230/LIPIcs.ICALP.2018.45}
  {\path{doi:10.4230/LIPIcs.ICALP.2018.45}}.

\bibitem[DMRW07]{demaine2007}
Erik~D. Demaine, Shay Mozes, Benjamin Rossman, and Oren Weimann.
\newblock An optimal decomposition algorithm for tree edit distance.
\newblock In Lars Arge, Christian Cachin, Tomasz Jurdzi{\'{n}}ski, and Andrzej
  Tarlecki, editors, {\em Automata, Languages and Programming}, pages 146--157,
  Berlin, Heidelberg, 2007. Springer Berlin Heidelberg.

\bibitem[DT03]{DulucqT03}
Serge Dulucq and H{\'{e}}l{\`{e}}ne Touzet.
\newblock Analysis of tree edit distance algorithms.
\newblock In {\em Proceedings of the 14th Annual Symposium on Combinatorial
  Pattern Matching ({CPM})}, volume 2676 of {\em Lecture Notes in Computer
  Science}, pages 83--95. Springer, 2003.
\newblock \href {https://doi.org/10.1007/3-540-44888-8\_7}
  {\path{doi:10.1007/3-540-44888-8\_7}}.

\bibitem[DT05]{DulucqT05}
Serge Dulucq and H{\'{e}}l{\`{e}}ne Touzet.
\newblock Decomposition algorithms for the tree edit distance problem.
\newblock {\em J. Discrete Algorithms}, 3(2-4):448--471, 2005.
\newblock \href {https://doi.org/10.1016/j.jda.2004.08.018}
  {\path{doi:10.1016/j.jda.2004.08.018}}.

\bibitem[FLMM09]{FerraginaLMM09}
Paolo Ferragina, Fabrizio Luccio, Giovanni Manzini, and S.~Muthukrishnan.
\newblock Compressing and indexing labeled trees, with applications.
\newblock {\em J. {ACM}}, 57(1):4:1--4:33, 2009.
\newblock \href {https://doi.org/10.1145/1613676.1613680}
  {\path{doi:10.1145/1613676.1613680}}.

\bibitem[GPVX21]{gu2021faster}
Yuzhou Gu, Adam Polak, Virginia {Vassilevska Williams}, and Yinzhan Xu.
\newblock Faster monotone min-plus product, range mode, and single source
  replacement paths, 2021.
\newblock To appear in ICALP 2021.
\newblock \href {http://arxiv.org/abs/2105.02806} {\path{arXiv:2105.02806}}.

\bibitem[Gus97]{gusfield_1997}
Dan Gusfield.
\newblock {\em Algorithms on Strings, Trees, and Sequences: Computer Science
  and Computational Biology}.
\newblock Cambridge University Press, 1997.
\newblock \href {https://doi.org/10.1017/CBO9780511574931}
  {\path{doi:10.1017/CBO9780511574931}}.

\bibitem[HT84]{doi:10.1137/0213024}
Dov Harel and Robert~Endre Tarjan.
\newblock Fast algorithms for finding nearest common ancestors.
\newblock {\em SIAM Journal on Computing}, 13(2):338--355, 1984.
\newblock \href {https://doi.org/10.1137/0213024} {\path{doi:10.1137/0213024}}.

\bibitem[HTGK03]{HochsmannTGK03}
Matthias H{\"{o}}chsmann, Thomas T{\"{o}}ller, Robert Giegerich, and Stefan
  Kurtz.
\newblock Local similarity in {RNA} secondary structures.
\newblock In {\em Proceedings of 2nd {IEEE} Computer Society Bioinformatics
  Conference, {CSB}}, pages 159--168. {IEEE} Computer Society, 2003.
\newblock \href {https://doi.org/10.1109/CSB.2003.1227315}
  {\path{doi:10.1109/CSB.2003.1227315}}.

\bibitem[Kle98]{Klein98}
Philip~N. Klein.
\newblock Computing the edit-distance between unrooted ordered trees.
\newblock In {\em Proceedings of the 6th Annual European Symposium on
  Algorithms ({ESA})}, volume 1461 of {\em Lecture Notes in Computer Science},
  pages 91--102. Springer, 1998.
\newblock \href {https://doi.org/10.1007/3-540-68530-8\_8}
  {\path{doi:10.1007/3-540-68530-8\_8}}.

\bibitem[KS20]{KouckyS20}
Michal Kouck{\'{y}} and Michael~E. Saks.
\newblock Constant factor approximations to edit distance on far input pairs in
  nearly linear time.
\newblock In {\em Proccedings of the 52nd Annual {ACM} {SIGACT} Symposium on
  Theory of Computing {(STOC)}}, pages 699--712. {ACM}, 2020.
\newblock \href {https://doi.org/10.1145/3357713.3384307}
  {\path{doi:10.1145/3357713.3384307}}.

\bibitem[KSK01]{KleinSK01}
Philip~N. Klein, Thomas~B. Sebastian, and Benjamin~B. Kimia.
\newblock Shape matching using edit-distance: an implementation.
\newblock In {\em Proceedings of the 12th Annual Symposium on Discrete
  Algorithms ({SODA})}, pages 781--790, 2001.
\newblock URL: \url{http://dl.acm.org/citation.cfm?id=365411.365779}.

\bibitem[KTSK00]{KleinTSK00}
Philip~N. Klein, Srikanta Tirthapura, Daniel Sharvit, and Benjamin~B. Kimia.
\newblock A tree-edit-distance algorithm for comparing simple, closed shapes.
\newblock In David~B. Shmoys, editor, {\em Proceedings of the 11th Annual
  {ACM-SIAM} Symposium on Discrete Algorithms ({SODA})}, pages 696--704, 2000.
\newblock URL: \url{http://dl.acm.org/citation.cfm?id=338219.338628}.

\bibitem[LG12]{legall2012}
Francois Le~Gall.
\newblock Faster algorithms for rectangular matrix multiplication.
\newblock FOCS '12, page 514–523, USA, 2012. IEEE Computer Society.
\newblock \href {https://doi.org/10.1109/FOCS.2012.80}
  {\path{doi:10.1109/FOCS.2012.80}}.

\bibitem[LV88]{LandauV88}
Gad~M. Landau and Uzi Vishkin.
\newblock Fast string matching with k differences.
\newblock {\em J. Comput. Syst. Sci.}, 37(1):63--78, 1988.
\newblock \href {https://doi.org/10.1016/0022-0000(88)90045-1}
  {\path{doi:10.1016/0022-0000(88)90045-1}}.

\bibitem[MP80]{MasekP80}
William~J. Masek and Mike Paterson.
\newblock A faster algorithm computing string edit distances.
\newblock {\em J. Comput. Syst. Sci.}, 20(1):18--31, 1980.
\newblock \href {https://doi.org/10.1016/0022-0000(80)90002-1}
  {\path{doi:10.1016/0022-0000(80)90002-1}}.

\bibitem[Mye86]{Meyers86}
Eugene~W. Myers.
\newblock An {O(ND)} difference algorithm and its variations.
\newblock {\em Algorithmica}, 1(2):251--266, 1986.
\newblock \href {https://doi.org/10.1007/BF01840446}
  {\path{doi:10.1007/BF01840446}}.

\bibitem[Sch98]{doi:10.1137/S0097539795288489}
Jeanette~P. Schmidt.
\newblock All highest scoring paths in weighted grid graphs and their
  application to finding all approximate repeats in strings.
\newblock {\em SIAM Journal on Computing}, 27(4):972--992, 1998.
\newblock \href
  {http://arxiv.org/abs/https://doi.org/10.1137/S0097539795288489}
  {\path{arXiv:https://doi.org/10.1137/S0097539795288489}}, \href
  {https://doi.org/10.1137/S0097539795288489}
  {\path{doi:10.1137/S0097539795288489}}.

\bibitem[SE83]{SLEATOR1983362}
Daniel~D. Sleator and Robert {Endre Tarjan}.
\newblock A data structure for dynamic trees.
\newblock {\em Journal of Computer and System Sciences}, 26(3):362--391, 1983.
\newblock \href {https://doi.org/https://doi.org/10.1016/0022-0000(83)90006-5}
  {\path{doi:https://doi.org/10.1016/0022-0000(83)90006-5}}.

\bibitem[Sel77]{selkow77}
Stanley~M. Selkow.
\newblock The tree-to-tree editing problem.
\newblock {\em Information Processing Letters}, 6(6):184--186, 1977.
\newblock \href {https://doi.org/https://doi.org/10.1016/0020-0190(77)90064-3}
  {\path{doi:https://doi.org/10.1016/0020-0190(77)90064-3}}.

\bibitem[SKK04]{SebastianKK04}
Thomas~B. Sebastian, Philip~N. Klein, and Benjamin~B. Kimia.
\newblock Recognition of shapes by editing their shock graphs.
\newblock {\em {IEEE} Trans. Pattern Anal. Mach. Intell.}, 26(5):550--571,
  2004.
\newblock \href {https://doi.org/10.1109/TPAMI.2004.1273924}
  {\path{doi:10.1109/TPAMI.2004.1273924}}.

\bibitem[SPA17]{schwarz2017}
Stefan Schwarz, Mateusz Pawlik, and Nikolaus Augsten.
\newblock A new perspective on the tree edit distance.
\newblock In Christian Beecks, Felix Borutta, Peer Kr{\"o}ger, and Thomas
  Seidl, editors, {\em Similarity Search and Applications}, pages 156--170,
  Cham, 2017. Springer International Publishing.

\bibitem[SZ90]{10.1093/bioinformatics/6.4.309}
Bruce~A. Shapiro and Kaizhong Zhang.
\newblock {Comparing multiple RNA secondary structures using tree comparisons}.
\newblock {\em Bioinformatics}, 6(4):309--318, 10 1990.
\newblock \href {https://doi.org/10.1093/bioinformatics/6.4.309}
  {\path{doi:10.1093/bioinformatics/6.4.309}}.

\bibitem[Tai79]{Tai79}
Kuo{-}Chung Tai.
\newblock The tree-to-tree correction problem.
\newblock {\em J. {ACM}}, 26(3):422--433, 1979.
\newblock \href {https://doi.org/10.1145/322139.322143}
  {\path{doi:10.1145/322139.322143}}.

\bibitem[Tou05]{Touzet05}
H{\'{e}}l{\`{e}}ne Touzet.
\newblock A linear tree edit distance algorithm for similar ordered trees.
\newblock In {\em Proceedings of the 16th Annual Symposium on Combinatorial
  Pattern Matching ({CPM})}, volume 3537 of {\em Lecture Notes in Computer
  Science}, pages 334--345. Springer, 2005.
\newblock \href {https://doi.org/10.1007/11496656\_29}
  {\path{doi:10.1007/11496656\_29}}.

\bibitem[Tou07]{Touzet07}
H{\'{e}}l{\`{e}}ne Touzet.
\newblock Comparing similar ordered trees in linear-time.
\newblock {\em J. Discrete Algorithms}, 5(4):696--705, 2007.
\newblock \href {https://doi.org/10.1016/j.jda.2006.07.002}
  {\path{doi:10.1016/j.jda.2006.07.002}}.

\bibitem[Ukk85]{Ukkonen85}
Esko Ukkonen.
\newblock Algorithms for approximate string matching.
\newblock {\em Inf. Control.}, 64(1-3):100--118, 1985.
\newblock \href {https://doi.org/10.1016/S0019-9958(85)80046-2}
  {\path{doi:10.1016/S0019-9958(85)80046-2}}.

\bibitem[{Vas}18]{vvwsurvey}
Virginia {Vassilevska Williams}.
\newblock On some fine-grained questions in algorithms and complexity.
\newblock page 3447–3487. International Congress of Mathematicians (ICM),
  2018.
\newblock \href {https://doi.org/10.1142/9789813272880_0188}
  {\path{doi:10.1142/9789813272880_0188}}.

\bibitem[VX]{xu2020}
Virginia {Vassilevska Williams} and Yinzhan Xu.
\newblock Truly subcubic min-plus product for less structured matrices, with
  applications.
\newblock In {\em Proceedings of the 2020 ACM-SIAM Symposium on Discrete
  Algorithms (SODA)}, pages 12--29.
\newblock \href {https://doi.org/10.1137/1.9781611975994.2}
  {\path{doi:10.1137/1.9781611975994.2}}.

\bibitem[Wat95]{waterman1995introduction}
Michael~S. Waterman.
\newblock {\em Introduction to computational biology: maps, sequences and
  genomes}.
\newblock CRC Press, 1995.

\bibitem[Wil14]{rrwapsp2014}
Ryan Williams.
\newblock Faster all-pairs shortest paths via circuit complexity.
\newblock STOC '14, page 664–673, New York, NY, USA, 2014. Association for
  Computing Machinery.
\newblock \href {https://doi.org/10.1145/2591796.2591811}
  {\path{doi:10.1145/2591796.2591811}}.

\bibitem[ZS89]{ZhangS89}
Kaizhong Zhang and Dennis~E. Shasha.
\newblock Simple fast algorithms for the editing distance between trees and
  related problems.
\newblock {\em {SIAM} J. Comput.}, 18(6):1245--1262, 1989.
\newblock \href {https://doi.org/10.1137/0218082} {\path{doi:10.1137/0218082}}.

\end{thebibliography}

\newpage
\appendix

\section{Similarity matrices are not Anti-Monge matrices} \label{appendix:notmonge}

We show that the similarity matrices defined in Definition \ref{def:similaritymatrix} do not admit the Anti-Monge property, which is defined as follows:

\begin{definition} [Anti-Monge property] \label{def:monge}
An $n \times n$ matrix $A = a_{ij}$ has the \textit{Anti-Monge Property} if and only if for all $1 \le i \le i ^ {\prime} \le n$ and $1 \le j \le j ^ {\prime} \le n$, we have
\begin{align*}
    a_{ij} + a_{i ^ {\prime}j ^ {\prime}} \ge a_{ij ^ {\prime}} + a_{i ^ {\prime}j}.
\end{align*}
\end{definition}

Consider the following counter example shown in Figure \ref{fig:counterexample} (the lowercase letters on nodes are the symbols, and the subscripts are the indices):

\begin{figure}[H] 
    \centering 
    \begin{tikzpicture}[
scale=0.5,
every node/.style={solid},
edge from parent/.style={draw,solid}, 
level/.style={sibling distance=20mm},
level distance=30mm
]
\node [circle,draw,label=above:{$T_1$}] at (0, 0) {$a_1$}
    [child anchor=north]
    child {node [circle,draw] {$a_2$} edge from parent[solid]
        child {node [circle,draw] {$b_3$} edge from parent[solid]
        }
        child {node [circle,draw] {$c_4$} edge from parent[solid]
        }
        child {node [circle,draw] {$c_5$} edge from parent[solid]
        };
    };
    
\node [circle,draw,label=above:{$T_2$}] at (8, 0) {$d_1$}
    [child anchor=north]
    child {node [circle,draw] {$b_2$} edge from parent[solid]
    }
    child {node [circle,draw] {$a_3$} edge from parent[solid]
        child {node [circle,draw] {$a_4$} edge from parent[solid]
        };
    }
    child {node [circle,draw] {$c_5$} edge from parent[solid]
    }
    child {node [circle,draw] {$c_6$} edge from parent[solid]
    };
    
\end{tikzpicture}
\caption{A counter-example for the Anti-Monge Property} \label{fig:counterexample}
\end{figure}
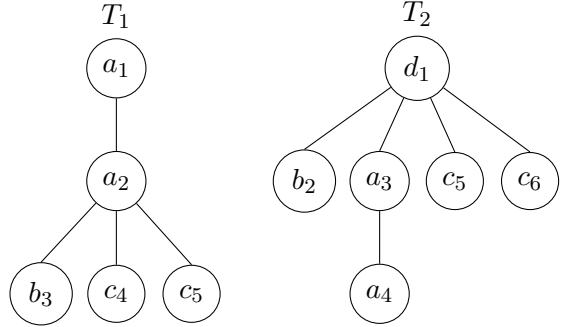  

We now argue that $S(T_1, T_2)$ does not admit the Anti-Monge property. The bi-order traversal sequence of $T_2$ is $d_1b_2b_2a_3a_4a_4a_3c_5c_5c_6c_6d_1$. Therefore, we have
\begin{itemize}
    \item $T_2[2, 8)$ consists of nodes $\{b_2, a_3, a_4\}$ and $\similarity(T_1, T_2[2, 8)) = 4$ with the mapping from $T_1$ to $T_2$ being $\{a_1 \rightarrow a_3, a_2 \rightarrow a_4\}$;
    \item $T_2[4, 12)$ consists of nodes $\{a_3, a_4, c_5, c_6\}$ and $\similarity(T_1, T_2[4, 12)) = 5$ with the mapping from $T_1$ to $T_2$ being $\{b_3 \rightarrow a_3, c_4 \rightarrow c_5, c_5 \rightarrow c_6\}$;
    \item $T_2[2, 12)$ consists of nodes $\{b_2, a_3, a_4, c_5, c_6\}$ and $\similarity(T_1, T_2[2, 12)) = 6$ with the mapping from $T_1$ to $T_2$ being $\{b_3 \rightarrow b_2, c_4 \rightarrow c_5, c_5 \rightarrow c_6\}$;
    \item $T_2[4, 8)$ consists of nodes $\{a_3, a_4\}$consists of node $b_3$ and $\similarity(T_1, T_2[4, 8)) = 4$ with the mapping from $T_1$ to $T_2$ being $\{a_1 \rightarrow a_3, a_2 \rightarrow a_4\}$.
\end{itemize}
Therefore, we have
\begin{align*}
    \similarity(T_1, T_2[2, 8)) + \similarity(T_1, T_2[4, 12)) = 4 + 5 < 6 + 4 = \similarity(T_1, T_2[2, 12)) + \similarity(T_1, T_2[4, 8))
\end{align*}
and $(i, i ^ {\prime}, j, j ^ {\prime}) = (2, 4, 8, 12)$ gives us a counter-example to the property in Definition \ref{def:monge} for $S(T_1, T_2)$.

\end{document}